\documentclass[10pt,onecolumn,draftcls]{IEEEtran}
\ifCLASSINFOpdf
\else
\fi
\usepackage{cite}
\usepackage{amsmath}
\usepackage{graphicx} 
\usepackage{amsthm}
\usepackage{algpseudocode}
\usepackage{algorithm}
\usepackage{epstopdf}
\usepackage{latexsym,bm,amssymb,amsfonts}%
\usepackage{enumitem}
\usepackage{booktabs} 
\usepackage{color}
\usepackage[format=plain]{caption}
\usepackage{subfigure}
\newtheorem{lem}{Lemma}
\newtheorem{prop}{Proposition}

\newtheorem{rmk}{Remark}

\theoremstyle{plain}

\definecolor{orange}{RGB}{255,107,0}

\def\black{\color{black}}

\begin{document}
	%
	\title{A Penalized Inequality-Constrained Approach for Robust Beamforming with DoF Limitation}

	\author{Wenqiang Pu, Jinjun Xiao, Tao Zhang,  Zhi-Quan Luo
		\thanks{Part of this work is presented in IEEE WASPAA, Oct. 2017~\cite{Pu2017penalized}, IEEE ICASSP, April 2018~\cite{Xiao2018Evaluation}, and IEEE SAM, June 2020~\cite{Conformal2020}. This work is partially supported by a research gift from Starkey Hearing Technologies. The work of Wenqiang Pu is supported by the National Natural Science Foundation of China (No. 62101350). The work of Z.-Q. Luo is supported by the National Natural Science Foundation of China (No. 61731018) and the Guangdong Provincial Key Laboratory of Big Data Computing.}
		\thanks{J. Xiao and T. Zhang are with Starkey Hearing Technologies, Minneapolis, MN 55455, USA. Email: $\{$jinjun$\_$xiao,tao$\_$zhang$\}$@starkey.com}
		\thanks{W. Pu and Z.-Q. Luo are with Shenzhen Research Institute of Big Data, The Chinese University of Hong Kong, Shenzhen, 518172, China. E-mail: $\{$wenqiangpu,luozq$\}$@cuhk.edu.cn}}
	
	\maketitle
	
	\begin{abstract}
		A well-known challenge in beamforming is how to optimally utilize the degrees of freedom (DoF) of the array to design a robust beamformer, especially when the array DoF is limited. In this paper, we leverage the tool of constrained convex optimization and propose a penalized inequality-constrained minimum variance (P-ICMV) beamformer to address this challenge. Specifically, a well-targeted objective function and inequality constraints are proposed to achieve the design goals. By penalizing the maximum gain of the beamformer at any interfering directions, the total interference power can be efficiently mitigated with limited DoF. 
		Multiple robust constraints on the target protection and interference suppression can be introduced to increase the robustness of the beamformer against steering vector mismatch. By integrating the noise reduction, interference suppression, and target protection, the proposed formulation can efficiently obtain a robust beamformer design while optimally trading off  various design goals. $\ $To numerically solve this problem, we formulate the P-ICMV beamformer design as a convex second-order cone program (SOCP) and propose a low complexity iterative algorithm based on the alternating direction method of multipliers (ADMM). Three applications are simulated to demonstrate the effectiveness of the proposed beamformer.
	\end{abstract}
	

	%
	\IEEEpeerreviewmaketitle

	\section{Introduction}\label{sec:intro}
	
	Beamforming is a fundamental technique in array signal processing, which exploits the {\black spatial diversity} to enhance the desired signal and {\black suppresses} undesired interferences and noise. 
	The beamforming technique has been widely used in many multi-channel signal processing areas, {\black e.g.}, wireless communication \cite{tse2005fundamentals}, microphone array speech processing \cite{Doclo2015Magzine}, radar \cite{ho2013radar}, sonar \cite{chiang2005sonar}, medical imaging \cite{Synnevag2009Benefits}, etc. 
	The key procedure in beamforming is specifying the so called \textit{beamformer}, which serves as a vector of complex coefficients to linearly combine signals received by array elements. 
	In the past decades, various beamformer design criteria have been extensively studied, these criteria can be divided into two types: data-independent and data-dependent criteria. 
	The performance of data-independent beamformers~\cite{elko1996microphone,mabande2009design,Zhang2017ARC} is limited {\black since useful information of the signal environment is not exploited.} Instead, the data-dependent beamformers can deliver optimal performance due to their adaptivity to the signal environment and, as such, they are also referred as adaptive beamformers. One representative adaptive beamformer is the minimum variance distortionless response (MVDR) beamformer~\cite{capon1969high}, which has been widely used in {\black many applications due to the convenience of practical implementation.}
	
	However, the performance of the MVDR beamformer suffers a degradation due to the imperfect knowledge of {\black the environment}, {\black e.g.,} imprecise {steering vectors} (SV), finite number of snapshots, direction-of-arrival (DoA) errors, etc. 
	Generally speaking, these imperfections lead to two types of model uncertainties of the MVDR beamformer.
	One uncertainty is the mismatch of SV of the target signal, {\black which is caused by array calibration error and DoA error.} This {\black would} lead to inevitable distortion of the target signal. 
	The other uncertainty {\black is} the estimation error of the covariance matrix of interference plus noise. {\black It is caused by the finite number of signal samples, the presence of the target signal in the training samples, and the non-stationarity of signals.} Such uncertainty not only {\black leads to  performance degradation on interference suppression}, but also distorts the target signal if it is present in the training samples.
	Faced with the above two types of uncertainties, many robust beamformers were proposed and studied in the past decades, to mitigate the negative effects of one or both of the uncertainties. 
	
	A comprehensive review of the principles for robust adaptive beamforming technique {\black can be found} in~\cite{VOROBYOV20133264}. 
	Here after, we briefly review several representative robust beamforming techniques. 
	
	To handle the uncertainty of SV, the linearly constrained minimum variance (LCMV) beamformer~\cite{Buckley1987LCMV} enforces multiple distortionless constraints on possible DoAs of the target signal. Though the LCMV beamformer stabilizes the mainlobe response for the target, it does use up the DoF resource that could be otherwise used for suppressing interference and noise. When the SVs enforced in constraints contain DoA errors, a recent work~\cite{chakrabarty2017bayesian} proposes a Bayesian approach to jointly estimate DoAs and suppresses interference and noise. Instead of enforcing multiple equality constraints, uncertainty set based beamformers~\cite{Vorobyov2003worst,Li2003Capon,Lorenz2005Robust,Nai2011Iterative,Huang2018,liao2017robust} assume the mismatch of SV lies in a bounded spherical or ellipsoidal. {\black The mainlobe can be stabilized by enforcing a worst-case criterion~\cite{Vorobyov2003worst,Li2003Capon,Lorenz2005Robust} based constraint} or by iteratively estimating the SV from the uncertainty set~\cite{Nai2011Iterative}. Another way for mitigating the uncertainty of SV is the eigenspace based beamformers~\cite{Chang1992eigenspace,Feldman1996projection,HUANG20121758}, which modify the SV by projecting it onto the estimated signal-plus-interference subspace. Besides, {\black authors in~\cite{ruan2016robust}} propose an orthogonal Krylov subspace-based method to estimate the SV in a reduced-dimensional subspace. These subspace-based beamformers may suffer from noise corrosion when signal-to-noise ratio (SNR) is low. 
	
	To handle the uncertainty of the covariance matrix, several robust techniques are proposed by reconstructing the sample covariance matrix. The well known one is the diagonal loading (DL) based beamformer~\cite{Cox1987Robust,Carlson1988Covariance,Elnashar2006Diagonal}, which modifies the sample covariance matrix by adding a diagonal matrix. {\black Though the DL beamformer  only modifies the covariance matrix, it also enables to mitigate the impact of SV mismatch~\cite{Li2003Capon}}. However, its performance may degrade in high signal-to-interference plus noise ratio (SINR) situation if the target signal is present in training samples. By exploiting the a prior knowledge of the array manifold, several covariance matrix reconstruction based beamformers were proposed~\cite{Gu2012Reconstruction,Ruan2014Shrinkage,Huang2015Reconstruction,zhang2015interference}. These beamformers reconstruct the interference-plus-noise covariance matrix and estimate the SV of the target signal from presumed SVs. Note that the covariance matrix reconstruction procedure are based on integration operation over specified angle regions, which may require large computational cost especially for a large size array. {\black Authors in~\cite{yang2018high} propose a spiked random matrix based model to reconstruct it.} But this work does not consider the SV mismatch issue.

	In this paper, we consider the robust beamformer design and assume that the \textit{presumed} SVs for different directions are available. Also, we assume a rough DoA estimation for each source can be obtained by a suitable DoA estimation algorithm. These basic assumptions can be satisfied in many beamforming applications such as radar~\cite{ho2013radar} and microphone array~\cite{Doclo2015Magzine}. By exploiting these prior knowledge, we propose a penalized inequality-constraint minimum variance (P-ICMV) beamformer {\black formulation} and also develop an efficient optimization algorithm to {\black solve the proposed formulation}. Our main contributions are:
	\begin{enumerate}[fullwidth,label=(\arabic*)]
		
		\item \textbf{Robustness against three kinds of errors:} The DoA error, the SV mismatch error, and the covariance matrix estimation error, which usually appear in practical scenarios are simultaneously considered in the robust beamformer design. {\black Robustness against these errors is achieved by enforcing multiple inequality constraints, in which the presumed SVs are used to control the spatial responses of directions of interests.}

		\item \textbf{A min-max penalization criterion:} Enforcing multiple inequality constraints with limited DoF faces a design feasibility issue, especially in a multi-source scenario. To allow a feasible beamformer design under various robustness considerations, a penalization criterion is proposed to intelligently allocate the limited array DoF. Specifically, the spatial responses on all possible directions of interferences are penalized, this makes the P-ICMV beamformer able to handle multiple sources with robustness.
		
		\item \textbf{An efficient optimization algorithm:} The proposed P-ICMV formulation is a convex second order cone programing (SOCP), we develop a low complexity iterative algorithm based on the alternating direction method of multipliers (ADMM) algorithm to solve it. {\black Also, the developed optimziation technique for solving a class of complex SOCPs itself is interesting.}
		
	\end{enumerate}
	
	We adopt the following notations in this paper. Lower and upper case letters in bold are used for vectors and matrices respectively. For a given matrix $\mathbf{X},$ we denote its transpose and Hermitian transpose by $\mathbf{X}^T$ and $\mathbf{X}^H$ respectively. If $\mathbf{X}$ is a square matrix, we use $\lambda_{\textrm{min}}(\mathbf{X})$ and $\lambda_{\textrm{max}}(\mathbf{X})$ to denote its smallest and largest eigenvalue respectively, use $\mathbf{X}\succ \mathbf{0}$ ($\mathbf{X}\succeq \mathbf{0}$) to denote that it is a positive definite (semi-definite) matrix, and use $\mathbf{X}^{-1}$ to denote its inverse (if exists). For a given complex number $c$, we use $\textrm{Re}\{c\}$ and $\textrm{Im}\{c\}$ to denote its real part and imaginary part respectively, and use $\angle c$ to denote the angle of $c$. We use $\|\mathbf{x}\|$ to denote Euclidean ($
	\ell_2$) norm of the vector $\mathbf{x}$. $\mathbb{E}\left[\cdot\right]$ is used to represent the expectation operation. The notation $\mathbf{I}$ represents the identity matrix with an appropriate size, and the $M$-dimensional real (complex) vector space is denoted by $\mathbb{R}^{M}$ ($\mathbb{C}^{M}$).
	
	
	\section{Signal Model}
	Consider an array with $M$ {\black omnidirectional} array elements. There are $K+1$ statistically independent {\black narrowband} signal sources at different directions $\theta_k,\ k=0,1,\ldots,K$ $(\theta_k\neq \theta_{k^\prime},\ \forall k\neq k^\prime)$.  The signal received by the array at time instances $n=1,2,\ldots,$ {\black is modeled as}
	\begin{equation}\label{eq:sig_model}
	\mathbf{x}(n)=s_0(n)\mathbf{a}_{\theta_0}+ \sum\nolimits_{k=1}^Ks_k(n)\mathbf{a}_{\theta_k}+\mathbf{v}(n)\in\mathbb{C}^M,
	\end{equation}
	where $s_k(n)\in\mathbb{C}$ is the $k$-th source signal at discrete time instance $n$,  $\mathbf{a}_{\theta_k}\in\mathbb{C}^M$ is the SV of direction $\theta_k$, and $\mathbf{v}(n)\in\mathbb{C}^M$ is the noise. {\black Let $s_0(n)$ be} the target signal and others be unwanted signal, then \eqref{eq:sig_model} can be simplified as 
	\begin{equation}\label{eq:sig_model_simp}
	\mathbf{x}(n)=s_0(n)\mathbf{a}_{\theta_0}+\mathbf{u}(n)\in\mathbb{C}^M,
	\end{equation}
	where $\mathbf{u}(n)=\sum_{k=1}^Ks_k(n)\mathbf{a}_{\theta_k}+\mathbf{v}(n)$ is the unwanted signal.
	The beamforming technique linearly combines $\mathbf{x}(n)$ by a so-called beamformer $\mathbf{w}\in\mathbb{C}^M$ such that the output signal 
	\begin{equation}
	z(n)=\mathbf{w}^H\mathbf{x}(n)
	\end{equation}
	satisfies a specified requirement, {\black e.g., SINR of $z(n)$ is maximized.}
	
	With the \textit{a priori} knowledge of the target signal's SV $\mathbf{a}_{\theta_0}$, the well-known MVDR beamformer attempts to minimize the undesired signal power at beamforming output subject to a distortionless constraint on the array response at direction $\theta_0$. The MVDR beamformer formulation is
	\begin{subequations}\label{eq:opt_mvdr}
		\begin{align}
		\min_{\mathbf{w}}\ &\mathbb{E}\left[  |\mathbf{w}^H\mathbf{u}(n) |^2   \right] \label{eq:mvdr_obj}\\
		\textrm{s.t.}\ &\mathbf{w}^H\mathbf{a}_{\theta_0}=1.\label{eq:mvdr_cons}
		\end{align}
	\end{subequations}
	Linear constraint~\eqref{eq:mvdr_cons} guarantees $s_0(n)$ being preserved in $z(n)$ without any distortion and the objective function in~\eqref{eq:mvdr_obj} can be further expressed as
	\begin{equation}
	\mathbb{E}\left[  |\mathbf{w}^H\mathbf{u}(n) |^2   \right]=\mathbf{w}^H\mathbf{R}_u\mathbf{w},
	\end{equation} 
	where $\mathbf{R}_u\triangleq\mathbb{E}\left[ \mathbf{u}(n)\mathbf{u}^H(n) \right]$ is the covariance matrix of $\mathbf{u}(n)$. By the Lagrange multiplier method~\cite{bertsekas1999nonlinear}, the optimal solution of problem \eqref{eq:opt_mvdr} is
	\begin{equation}\label{eq:sulo_mvdr}
	\mathbf{w}^*=\frac{\mathbf{R}_u^{-1}\mathbf{a}_{\theta_0}}{\mathbf{a}^H_{\theta_0}\mathbf{R}_u^{-1}\mathbf{a}_{\theta_0}}.
	\end{equation}
	
	Though the MVDR beamformer $\mathbf{w}^*$ theoretically guarantees a maximal SINR at the beamforming output~\cite{van2004optimum}, its performance could degrade significantly in practice due to:
	\begin{enumerate}[fullwidth,label=(\arabic*)]
		\item Inaccuracy of $\mathbf{R}_u$: The exact covariance matrix $\mathbf{R}_u$ is usually unavailable in most applications, and it usually requires to be estimated from finite $N$ training signal samples. Several covariance matrix estimation approaches~\cite{Cox1987Robust,Carlson1988Covariance,Elnashar2006Diagonal,Gu2012Reconstruction,Ruan2014Shrinkage,Huang2015Reconstruction,zhang2015interference,yang2018high} were proposed to reduce the estimation error or  improve the robustness against the imprecise knowledge of the statistic property of the undesired signal $\mathbf{u}(n)$. 
		\item Mismatch of $\mathbf{a}_{\theta_0}$: Linear constraint~\eqref{eq:mvdr_cons} {\black is enforced to preserve $s_0(n)$ in $z(n)$ without distortion.} In practice, $\mathbf{a}_{\theta_0}$ can be inaccurate and the mismatch of $\mathbf{a}_{\theta_0}$ would lead to undesirable target signal distortion. Further, {\black this mismatch together with the appearance of target signal in the training signal samples would lead a heavy target signal distortion, which is the so-called signal cancellation phenomenon~\cite{Widrow1982cancellation}.}
		
	\end{enumerate}

	Many useful techniques have been developed in the past decades to deal with the mismatch of $\mathbf{a}_{\theta_0}$. But robustness against the inaccuracy of $\mathbf{R}_u$ is usually achieved via estimation or reconstruction. Another promising path is paved by enforcing additional constraints for interference suppression. However, this approach has not been well studied so far. The major challenge lies in the limited DoF of the array. Specifically, for linear constraints-based beamformers~\cite{marquardt2015interaural,Hadad2016Theoretical}, e.g., the well-known LCMV, the number of linear constraints, in general, should be no more than the number of array elements. Otherwise, the associated optimization problem becomes infeasible. Such a limitation (on the number of linear constraints) blocks the further exploration of using linear constraints for interference suppression to achieve robustness.
	
	
	On the other hand, recent researches~\cite{Liao2015bcd,Liao2016admm} relax the equality constraints into inequality ones, and a so-named inequality constrained minimum variance (ICMV) beamformer is proposed. Although the ICMV beamformer may also be limited by the number of constraints, {\black the flexibility of inequality constraints provides a potential way to address the DoF limitation issue.} In the next section, we will present our proposed P-ICMV beamformer formulation, which includes two types of inequality constraints and a penalized objective function. A key advantage of P-ICMV beamformer is its ability to effectively handle multiple sources with robustness on the best effort basis regardless of array DoF. Further, it also provides a flexible robustness adjustment to design beamformers for different situations.
	
	
	\begin{section}{Proposed P-ICMV Beamformer Design}
		To address the mismatch of SVs and inaccuracy of covariance matrix, we firstly introduce two types of inequality constraints in subsections \ref{subsec:rst_svmis} and \ref{subsec:rst_Rmis} respectively.  Then we propose a min-max criterion to deal with the {\black DoF limitation issue} in subsection \ref{subsec:limite_dof}, and the final P-ICMV beamformer formulation is given in subsection \ref{subsec:picmv}.
		
		\begin{subsection}{Robustness against Steering Vector Mismatch}\label{subsec:rst_svmis}
			To handle the SV mismatch, we propose the following inequality constraint to control signal distortion within a certain level. Specifically, for a given direction $\theta$, beamformer $\mathbf{w}$ is constrained as 
			\begin{equation}\label{eq:cons_rst_target_one}
			|\mathbf{w}^H\mathbf{a}_\theta-1|+\delta\| \mathbf{w} \|\leq c_\theta,
			\end{equation}
			where $c_\theta\geq 0$ is a user-defined tolerance threshold in direction $\theta$ and {\black $\delta\geq0$ is a pre-defined parameter for imposing robustness against SV mismatch.} Let $\mathbf{a}_\theta = \bar{\mathbf{a}}_\theta+\Delta\mathbf{a}_\theta,$ where $\bar{\mathbf{a}}_\theta$ is the true SV and $\Delta\mathbf{a}_\theta$ is the perturbation vector with $\|\Delta\mathbf{a}_\theta\|\leq\delta$. {\black Then, inequality \eqref{eq:cons_rst_target_one} actually controls the true signal distortion, given in the follow propostion.}
			{\black 
			\begin{prop}\label{prop:target_bound}
				Suppose $\mathbf{a}_\theta=\mathbf{\bar{a}}_\theta+\Delta\mathbf{a}_\theta$ with $\| \Delta\mathbf{a}_\theta\|\leq\delta$, if $\mathbf{w}$ satisfies~\eqref{eq:cons_rst_target_one}, then $|\mathbf{w}^H\mathbf{\bar{a}}_\theta-1|\leq c_\theta$.
			\end{prop}
			}
			\begin{proof}
				Denote $\phi$ be the phase of $(\mathbf{w}^H\mathbf{a}_\theta-1)$, then 
				\begin{subequations}\label{subeq:prop}
					\begin{align}
					|\mathbf{w}^H\mathbf{\bar{a}}_\theta-1|
					\leq&\max_{\|\Delta\mathbf{a}_\theta\|\leq \delta}|\mathbf{w}^H(\mathbf{a}_\theta-\Delta\mathbf{a}_\theta) -1|\label{subeq:prop1}\\
					=&\max_{\|\Delta\mathbf{a}_\theta\|\leq \delta} \left|\mathbf{w}^H\mathbf{a}_\theta-1-\mathbf{w}^H\Delta\mathbf{a}_\theta\right|\\
					=& \left|\mathbf{w}^H\mathbf{a}_\theta-1+e^{j\phi}\max_{\|\Delta\mathbf{a}_\theta\|\leq \delta}|\mathbf{w}^H\Delta\mathbf{a}_\theta|\right|\label{subeq:prop2}\\
					=&|\mathbf{w}^H\mathbf{a}_\theta-1|+\max_{\|\Delta\mathbf{a}_\theta\|\leq \delta} |\mathbf{w}^H\Delta\mathbf{a}_\theta|\\
					=&|\mathbf{w}^H\mathbf{a}_\theta-1|+\delta\| \mathbf{w} \|\leq c_\theta.
					\end{align}
				\end{subequations}
				In \eqref{subeq:prop}, \eqref{subeq:prop1} is due the worst-case criterion and \eqref{subeq:prop2} is because the maximum of  $\left|\mathbf{w}^H\mathbf{a}_\theta-1-\mathbf{w}^H\Delta\mathbf{a}_\theta\right|$ is achieved when $-\mathbf{w}^H\Delta\mathbf{a}_\theta$ has the same phase of $\mathbf{w}^H\mathbf{a}_\theta-1$ and $|\mathbf{w}^H\Delta\mathbf{a}_\theta|$ achieve maximum. Further, $\max_{\|\Delta\mathbf{a}_\theta\|\leq \delta} |\mathbf{w}^H\Delta\mathbf{a}_\theta|$ achieves maximum if $\Delta\mathbf{a}_\theta=\delta\frac{\mathbf{w}}{\|\mathbf{w}\|}$. This completes the proof.
			\end{proof}
			Proposition \ref{prop:target_bound} implies the true signal distortion is controlled by the pre-defined parameter $c_\theta$ in the worst-case sense, one remark for it is given below.
			\begin{rmk}[Relation to Other Worst-Case Criterion]
			    {\black For other worst-case criterion based constraints~\cite{Vorobyov2003worst,Li2003Capon,Lorenz2005Robust} in the literature, the true signal distortion can not be precisely controlled.} The recent work~\cite{liao2017robust} proposes nonconvex quadratic constraints to control the true signal distortion and~\eqref{eq:cons_rst_target_one} can be regarded as a shrunk version of these nonconvex constraints~\cite{Conformal2020}. To deal these nonconvex constraints, the semidefinite relaxation technique used in~\cite{liao2017robust} lifts the problem dimension from $M$ to $M^2/2$, and thus faces a computational cost issue when $M$ is large. {\black Instead, shrinking into a special convex form as~\eqref{eq:cons_rst_target} potentially allows efficient algorithm developement as presented in Section~\ref{sec:admm}.}
			\end{rmk}
			
			{\black To handle the DoA error, multiple inequality constraints at nearby angles of $\theta_0$ are introduced. This leads to the following constraints:
			\begin{equation}\label{eq:cons_rst_target}
			|\mathbf{w}^H\mathbf{a}_\theta-1| + \delta \| \mathbf{w} \| \leq c_\theta,\ \forall \theta\in\Theta,
			\end{equation}
			where $\Theta$ be a pre-specified discrete angle set, e.g., $\Theta=\theta_0+\{-5^{\circ},0^{\circ},5^{\circ} \}$. Inequality constraints~\eqref{eq:cons_rst_target} try to form a wide aperture beam around $\theta_0$ to robustly protect the target signal and the signal distortion for each $\theta\in\Theta$ is bounded by $c_\theta$.} 

		\end{subsection}
		
		\begin{subsection}{Robustness Against Inaccurate Estimation of $\mathbf{R}_u$}\label{subsec:rst_Rmis}
			There are two ways to impose robustness against the inaccurate estimation of $\mathbf{R}_u$. One is to exploit the inherent structure of $\mathbf{R}_u$ and develop efficient estimators for $\mathbf{R}_u$ from finite samples~\cite{Gu2012Reconstruction,Ruan2014Shrinkage,Huang2015Reconstruction,zhang2015interference,yang2018high}. {\black Another} way is to exploit the a prior knowledge of the presumed SVs and impose robust constraints for interference suppression. Following the second way, we propose to impose the following inequality constraints,
			\begin{equation}\label{eq:cons_rst_itf}
			|\mathbf{w}^H\mathbf{a}_\phi|+\delta\| \mathbf{w}\| \leq c_\phi,\ \forall \phi\in\Phi_k,\ \forall k,
			\end{equation}
			where $\Phi_k$ is a pre-defined discrete angle set for the interference $k$ based on the corresponding DoA estimation $\theta_k$, {\black e.g., $\Phi_k=\theta_k+\{-5^{\circ},0^{\circ},5^{\circ}\}$}, and $c_\phi\geq 0$ specifies an allowable amplification for angle $\phi$. Similar to Proposition \ref{prop:target_bound}, constraints \eqref{eq:cons_rst_itf} implies an upper bound for the true spatial response $\mathbf{w}^H\mathbf{\bar{a}}_\phi$, given in Proposition \ref{prop:intf_bound}.
			{\black 
			\begin{prop}\label{prop:intf_bound}
				Suppose $\mathbf{a}_\phi=\mathbf{\bar{a}}_\phi+\Delta\mathbf{a}_\phi$ with $\| \Delta\mathbf{a}_\phi\|\leq\delta$, if $\mathbf{w}$ satisfies~\eqref{eq:cons_rst_target_one}, then $|\mathbf{w}^H\mathbf{\bar{a}}_\phi|\leq c_\phi$.
			\end{prop}
			}
			\begin{proof}
				The proof for Proposition \ref{prop:intf_bound} is similar to Proposition \ref{prop:target_bound} and is omitted here.
			\end{proof}
		
			Constraints~\eqref{eq:cons_rst_itf} try to form a sufficiently wide and bounded null at nearby angles of $\phi_k$. Interferences thus {\black can be suppressed with} robustness regardless the estimation quality of $\mathbf{R}_u$. 
			
			
			The MVDR formulation~\eqref{eq:opt_mvdr} can be modified by replacing constraint~\eqref{eq:mvdr_cons} with constraints \eqref{eq:cons_rst_target} and \eqref{eq:cons_rst_itf}, leading to a robust beamformer design. Such beamformer (with $\delta=0$) was recently proposed and studied in~\cite{Liao2015bcd, Xiao2018Evaluation}, named as  ICMV beamformer in short. However, the ICMV beamformer also suffers from a design feasibility issue, i.e., the ICMV formulation may not be feasible when the number of constraints is larger than the number of array elements. In fact, the design of beamformer under limited DoF condition is not just a challenge for the ICMV, but also for other multiple constraints based beamformer design like the LCMV beamformer or other linear constraints based beamformers~\cite{marquardt2015interaural,Hadad2016Theoretical}. The key question is whether there is a way to impose arbitrary number of robust constraints when the number of array elements is fixed? In what follows, we provide a positive answer affirmatively by proposing a minmax robust beamforming formulation.
		\end{subsection}
		
		\begin{subsection}{Limited DoF}\label{subsec:limite_dof}
			In this subsection, we introduce a penalization criterion for constraints \eqref{eq:cons_rst_itf}. The proposed design criterion always enables a feasible beamformer design (under mild conditions, see Proposition \ref{prop:suffi_fea}) and {\black robustly} suppresses multiple interferences. From the physical meaning of constraints~\eqref{eq:cons_rst_itf}, we know they attempt to bound the spatial responses at several potential angles of interferences. Instead of specifying fixed level $c_\phi$ which would cause a design feasibility issue, we introduce an extra optimization variable $\bm{\epsilon}=(\epsilon_1,\epsilon_2,\ldots,\epsilon_K)\in\mathbb{R}^K$ on the right-hand sides of constraints~\eqref{eq:cons_rst_itf} to control the levels of spatial responses. We propose the following min-max optimization criterion for suppressing multiple interferences: 
			\begin{equation}\label{eq:opt_minmax}
			\begin{aligned}
			\min_{\mathbf{w},\bm{\epsilon}}\ & \max_{k}\{ \gamma_k \epsilon_k   \} + \textrm{Other objectives}\\
			\textrm{s.t.}\ & |\mathbf{w}^H\mathbf{a}_\phi|+\delta\| \mathbf{w} \|\leq \epsilon_k c_\phi,\ \forall \phi\in\Phi_k,\ \forall k,\\
			& \textrm{Other feasible constraints}.
			\end{aligned}
			\end{equation}
			In~\eqref{eq:opt_minmax}, $\{\gamma_k\}$ are user-specified interference suppression parameters, interference with larger $\gamma_k$ has higher priority to be suppressed. Suppose $c_\phi>0, \forall \phi\in\Phi_k,\forall k$, criterion \eqref{eq:opt_minmax} always enables a feasible design for $\mathbf{w}$ by optimizing the values of $\{ \epsilon_k \}$ large enough. We also remark that $\gamma_k$ is a user-defined parameter, which should be made proportional to the estimated interference power level. One straightforward way is specifying $\gamma_r$ by the estimated Capon spectrum \cite{capon1969high}, then proposed beamformer is able to adaptively allocate limited DoF to suppress interferences with different intention.
		\end{subsection}
		\begin{subsection}{Proposed P-ICMV Beamformer}\label{subsec:picmv}
			Finally, combine the discussions in Subsections \ref{subsec:rst_svmis}, \ref{subsec:rst_Rmis}, and \ref{subsec:limite_dof}, the proposed P-ICMV formulation is 
			\begin{equation}\label{eq:opt_picmv}
			\begin{aligned}
			\min_{\mathbf{w},\bm{\epsilon}}\ & \mathbf{w}^H\mathbf{R}\mathbf{w}+\mu\max_{k}\{ \gamma_k \epsilon_k   \}\\
			\textrm{s.t.}\ &|\mathbf{w}^H\mathbf{a}_\theta-1| + \delta \| \mathbf{w}\|\leq c_\theta,\ \forall \theta\in\Theta,\\
			& |\mathbf{w}^H\mathbf{a}_\phi| + \delta\|\mathbf{w} \|\leq \epsilon_k c_\phi,\ \forall \phi\in\Phi_k,\ \forall k.
			\end{aligned}
			\end{equation}
			In \eqref{eq:opt_picmv}, $\mathbf{R}$ can be the estimated noise-only covariance matrix or the estimated noise-plus-interference correlation matrix, and $\mu\geq0$ is a trade-off parameter for noise and interference suppression. Besides the beamformer variable $\mathbf{w}$, the P-ICMV formulation \eqref{eq:opt_picmv} has an \textit{extra} optimization variable $\bm{\epsilon}$ which makes the upper bound on $|\mathbf{w}^H\mathbf{a}_\phi|$ adjustable. The number of constraints for interference suppression is no longer limited by the DoF. The feasibility can be achieved by optimizing $\epsilon$ large enough. {\black Thus,} if constraints for angle $\theta\in\Theta$ are nonempty, the P-ICMV beamformer is able to generate a feasible design with bounded signal distortion for arbitrary number of interferences. To strictly ensure P-ICMV formulation \eqref{eq:opt_picmv} being able to handle arbitrary number of interferences with bounded distortion, we provide the following sufficient condition in  Proposition \ref{prop:suffi_fea}. 
			\begin{prop}\label{prop:suffi_fea}
				Let $\mathbf{A}\in\mathbb{C}^{M\times |\Theta|}$ be the matrix stacked by $\mathbf{a}_\theta,\forall \theta\in\Theta$, suppose $\textrm{rank}(\mathbf{A})=|\Theta|$ and $\delta \leq \min_{\theta}c_\theta/\sqrt{\mathbf{1}^H(\mathbf{A}^H\mathbf{A})^{-1}\mathbf{1}}$. Then problem \eqref{eq:opt_picmv} is always feasible.
			\end{prop}
			\begin{proof}
				Consider the value of $\|\mathbf{w}\|$ subject to $\mathbf{A}^H\mathbf{w}=\mathbf{b}$, where $\mathbf{b}\in \mathbb{C}^{|\Theta|}$ is a complex vector. By the method of Lagrange multipliers, $\mathbf{w}^*=\mathbf{A}(\mathbf{A}^H\mathbf{A})^{-1}\mathbf{b}$ achieves the minimal $\|\mathbf{w}\|$, given as $\| \mathbf{w}^*\|=\sqrt{\mathbf{b}^H(\mathbf{A}^H\mathbf{A})^{-1}\mathbf{b}}$. If exists $\mathbf{b}$ such that $|b_\theta-1| + \delta \sqrt{\mathbf{b}^H(\mathbf{A}^H\mathbf{A})^{-1}\mathbf{b}}\leq c_\theta,\ \forall \theta\in\Theta$, then $\mathbf{w}^*$ is a feasible solution for problem \eqref{eq:opt_picmv}. Setting $\mathbf{b}=\mathbf{1}$ completes the proof.
			\end{proof}
			
			In short, we summarize several interesting properties of the P-ICMV formulation \eqref{eq:opt_picmv} in the following remarks.
			
			\begin{rmk}[Bounded Spatial Responses]
				By Propositions \ref{prop:target_bound} and \ref{prop:intf_bound}, if $\|\Delta\mathbf{a}_\theta\|\leq\delta$, then any feasible solution $(\mathbf{w},\bm{\epsilon})$ of problem \eqref{eq:opt_picmv} has a bounded signal distortion, $1-c_\theta\leq |\mathbf{w}^H\mathbf{\bar{a}}_\theta|\leq 1+ c_\theta,\forall \theta\in\Theta$. Similarly, spatial response for each interference is softly bounded by an adjustable parameter $\epsilon_k$, $|\mathbf{w}^H\mathbf{\bar{a}}_\phi|\leq \epsilon_kc_\phi,\forall \phi \in \Phi_k$. {\black Further, if angle-dependent SV mismatch is considered, $\delta$ can be specified to $\delta_\theta$. The claims in Propositions \ref{prop:target_bound} and \ref{prop:intf_bound} still hold and the developed algorithm in Section~\ref{sec:admm} can also be applied to solve problem~\eqref{eq:opt_picmv}.}
			\end{rmk}
			
			\begin{rmk}[DoF Allocation]
				Penalty function $\mu\max_k\{ \gamma_k\epsilon_k \}$ enables P-ICMV beamformer to intelligently allocate DoF to suppress the intended interference with larger weight $\gamma_k$. This allows selective interference suppression, i.e.,  larger weight can be applied on interferences with higher degree of annoyance.
			\end{rmk}
			
			\begin{rmk}[Robustness Adjustment]
				{The P-ICMV formulation~\eqref{eq:opt_picmv} contains a set of user-specified parameters, this allows a flexible robust beamformer design. Specifically, different levels of DoA errors can be handled by specifying angle sets $\Theta$ and $\Phi_k,\forall k$; different allowable signal distortions can be controlled by tolerance parameters $c_\theta$ and $c_\phi$; and different levels of SV mismatch can be dealt by specifying $\delta$. Together with parameters $\mu$ and $\gamma_k$ in the penalty function, the total number of user specified parameters is $N_{\rm{u}}=2(| \Theta | + \sum_{k}|\Phi_k|)+K+2$.}
			\end{rmk}
			
			Finally, {\black we note that for a paticular robust beamforming task, these $N_u$ parameters can be properly specified by making use of the prior knowledge or the robustness level that the user prefers.} Also, with proper settings for these parameters, P-ICMV beamformer reduces to several existing beamformers, {\black e.g.,} MVDR, LCMV, and DL beamformers. Furthermore, though the P-ICMV formulation is proposed for adaptive beamforming purpose, it can also be used to design a specified static beam pattern by setting $\mathbf{R}=\mathbf{I}$, see examples in Section~\ref{subsec:synthesis}. Next, we develop an efficient algorithm to solve it.

		\end{subsection}
	\end{section}

	\section{ADMM Algorithm for P-ICMV Beamformer}\label{sec:admm}
	Problem \eqref{eq:opt_picmv} is actually a convex second-order cone program (SOCP), and the well-studied interior point method \cite{ye2011interior} can be used to solve it. However, the computational efficiency of the interior point method is too high in many adaptive beamforming applications such as in an embedded array system with limited processing~\cite{Doclo2015Magzine,Liao2015bcd,Liao2016admm}. Recent theoretical progress on the convergence and convergence analysis~\cite{Hong2017} for the ADMM algorithm {\black provides an alternative way to efficiently solve problem~\eqref{eq:opt_picmv}.} In this section, we first reformulate~\eqref{eq:opt_picmv} as a convex SOCP with smooth objective function and then solve it based the ADMM algorithm. 
	
	By introducing extra optimization variables $t\in\mathbb{R}$ and $y\in\mathbb{R}$, problem \eqref{eq:opt_picmv} is equivalently reformulated as 
	\begin{equation}\label{eq:opt_picmv_t}
	\begin{aligned}
	\min_{\mathbf{w},t,y}\ & \mathbf{w}^H\mathbf{R}\mathbf{w}+\mu t\\
	\textrm{s.t.}\ &|\mathbf{w}^H\mathbf{a}_\theta-1|+\delta y\leq c_\theta,\ \forall \theta\in\Theta,\\
	& |\mathbf{w}^H\mathbf{a}_\phi| + \delta y\leq t c_\phi/\gamma_k,\ \forall \phi\in\Phi_k,\ \forall k, \\
	&\|\mathbf{w}\| \leq y.
	\end{aligned}
	\end{equation}
	Notice that problems \eqref{eq:opt_picmv} and \eqref{eq:opt_picmv_t} have the same optimal $\mathbf{w}^*$, and the optimal $\epsilon_k^*, \forall k$ for problem \eqref{eq:opt_picmv} can be extracted from the optimal $(\mathbf{w}^*,y^*,t^*)$ for problem $\eqref{eq:opt_picmv_t}$, 
	given as $\max_{\phi\in\Phi_k}\{(|\mathbf{a}_\phi^H\mathbf{w}^*| + \delta y^*)/ c_\phi \}\leq\epsilon_k^*\leq t^*/\gamma_k$. To derive the ADMM algorithm for problem~\eqref{eq:opt_picmv_t}, we first introduce auxiliary variables $\{y_\theta,z_\theta\}$ and $\{y_\phi,z_\phi\}$ as
	\begin{subequations}
		\begin{align}
		y_\theta = y, \ z_\theta&= \mathbf{w}^H\mathbf{a}_{\theta},\ \forall \theta\in\Theta,\label{eq:delta_theta}\\
		y_\phi = y, \ z_\phi&= \mathbf{w}^H\mathbf{a}_{\phi},\ \forall \phi\in \Phi_k, \forall k\label{eq:delta_phi}.
		\end{align}
	\end{subequations}
	Then, problem~\eqref{eq:opt_picmv_t} can be equivalently reformulated as
	\begin{subequations}\label{eq:admm_picmv}
		\begin{align}
		\min\quad&\mathbf{w}^H\mathbf{R}\mathbf{w}+\mu t\notag \\
		\textrm{s.t.}\quad
		&|z_\theta-1|+\delta y_\theta\leq c_{\theta},\ \forall\theta\in\Theta,\label{eq:ricmv_tarineq}\\
		&|z_\phi|+\delta y_\phi\leq t c_{\phi}/\gamma_k,\ \forall\phi\in\Phi_k,\forall k,\label{eq:ricmv_intfineq}\\
		&\|\mathbf{w}\| \leq y,\label{eq:ricmv_wnorm}\\
		&\eqref{eq:delta_theta},\eqref{eq:delta_phi}.\notag
		\end{align}
	\end{subequations}
	
	Let $L_\rho(\mathbf{w},y,\{y_\theta,z_\theta\},t,\{y_\phi,z_\phi\},\{\eta_\theta,\lambda_\theta\},\{\eta_\phi,\lambda_\phi\})$ be the augmented Lagrangian function for problem~\eqref{eq:admm_picmv}~\cite{boyd2011distributed}
	\begin{equation*}
	\begin{aligned}
	L_\rho&(\mathbf{w},y,\{y_\theta,z_\theta\},t,\{y_\phi,z_\phi\},\{\eta_\theta,\lambda_\theta\},\{\eta_\phi,\lambda_\phi\})\\
    =	&\mathbf{w}^H\mathbf{R}\mathbf{w}+\mu t +\sum_{\theta\in\Theta}\left[ \textrm{Re}\{ \lambda_\theta^{ {H}}(\mathbf{w}^H{\mathbf{a}}_{\theta}-{z_\theta}) \}+\frac{\rho}{2}|\mathbf{w}^H{\mathbf{a}}_{\theta}-{z_\theta}|^2\right]\\
	&+\sum_{\theta\in\Theta}\left[  \eta_\theta(y-y_\theta) +\frac{\rho}{2}(y-y_\theta)^2\right]+\sum_k\sum_{\phi\in\Phi_k}\left[
	\textrm{Re}\{
	\lambda_{\phi}^{  {H}}(\mathbf{w}^H{\mathbf{a}}_{\phi}-{z_\phi})\}+\frac{\rho}{2}|\mathbf{w}^H{\mathbf{a}}_{\phi}-{z_\phi}|^2\right]\\
	&+\sum_k\sum_{\phi\in\Phi_k}\left[  \eta_\phi(y-y_\phi) +\frac{\rho}{2}(y-y_\phi)^2\right],
	\end{aligned}
	\end{equation*}
	where $\{\lambda_{\theta}\}$ and $\{\lambda_{\phi}\}$ are Lagrangian multipliers associated with equality constraints $z_\theta= \mathbf{w}^H\mathbf{a}_{\theta},\forall \theta\in\Theta$ and $z_\phi= \mathbf{w}^H\mathbf{a}_{\phi},\forall \phi\in \Phi_k, \forall k$, $\{\eta_{\theta}\}$ and $\{\eta_{\phi}\}$ are Lagrangian multipliers associated with equality constraints $y_\theta=y,\forall \theta\in\Theta$ and $y_\phi= y,\forall \phi\in \Phi_k, \forall k$, and $\rho>0$ is the penalty parameter in ADMM algorithm. Define $\mathbf{x}_{1}\triangleq(\mathbf{w},y)$, $\mathbf{x}_{2}\triangleq(\{y_\theta,z_\theta\})$, $\mathbf{x}_{3}\triangleq(t,\{y_\phi,z_\phi\})$, and $\bm{\lambda}=(\{\eta_\theta,\lambda_\theta\},\{\eta_\phi,\lambda_\phi\})$, 
	then at iteration $r=0,1,2,\cdots$, the ADMM algorithm updates all variables as follows:
	\begin{subequations}\label{eq:admm_update}
		\begin{align}
		&\mathbf{x}_1^{r+1}=\arg\min_{\eqref{eq:ricmv_wnorm}} L_\rho(\mathbf{x}_1,\mathbf{x}_2^r,\mathbf{x}_3^r,\bm{\lambda}^r),\label{eq:admm_w}\\
		&\mathbf{x}_2^{r+1}=\arg\min_{\eqref{eq:ricmv_tarineq}} L_\rho(\mathbf{x}_1^{r+1},\mathbf{x}_2,\mathbf{x}_3^r,\bm{\lambda}^r),\label{eq:admm_theta}\\
		&\mathbf{x}_3^{r+1}=\arg\min_{\eqref{eq:ricmv_intfineq}} L_\rho(\mathbf{x}_1^{r+1},\mathbf{x}_2^{r+1},\mathbf{x}_3,\bm{\lambda}^r),\label{eq:admm_phi}\\
		&\lambda_\theta^{r+1} = \lambda_\theta^{r} + \rho(\mathbf{w}^H{\mathbf{a}}_\theta-z_\theta^{r+1}),\theta\in\Theta, \label{eq:admm_lamtheta}\\
		&\eta_\theta^{r+1} = \eta_{\theta}^r + \rho(y-y_\theta),\theta\in\Theta, \label{eq:admm_etatheta}\\
		&\lambda_\phi^{r+1} = \lambda_\phi^{r} + \rho(\mathbf{w}^H{\mathbf{a}}_\phi-z_\phi^{r+1}),\phi\in\Phi_k,\forall k \label{eq:admm_lamphi},\\
		&\eta_\phi^{r+1} = \eta_\phi^{r} + \rho(y-y_\phi),\phi\in\Phi_k,\forall k \label{eq:admm_etaphi}.
		\end{align}
	\end{subequations}
	The convergence behavior of the above ADMM iterations is given in the following proposition~\cite{Hong2017}.
	\begin{prop}
		Suppose Proposition \ref{prop:suffi_fea} holds, then iterates $\{ \mathbf{w}^{r}\}$ generated by \eqref{eq:admm_update} converge to the optimal $\mathbf{w}^*$ of problem \eqref{eq:opt_picmv} as $r\rightarrow\infty$.
	\end{prop}
	
	Next, we derive closed-form solutions for subproblems \eqref{eq:admm_w}, \eqref{eq:admm_theta} and \eqref{eq:admm_phi} at each iteration $r$. We drop the iteration index $r$ for {\black notational} simplicity. For convenience of presenting the solutions for subproblems \eqref{eq:admm_theta} and \eqref{eq:admm_phi}, we first give Lemma \ref{lem:simp_QCQP_y}, which provides a closed-form solution for a special type of complex SOCP problem.
	\begin{lem}\label{lem:simp_QCQP_y}
		Consider a complex SOCP given as 
		\begin{equation}\label{eq:simp_QCQP_y}
		\min _{x,y} \ a|x|^2+{\rm{Re}}\{ b^Hx  \}+\alpha y^2 + \beta y\quad {\rm{s.t.}}\ |x-d|+\delta y\leq c,
		\end{equation}
		where $a,\alpha\in\mathbb{R}>0$, $c,\delta\in\mathbb{R}\geq0$, $\beta\in\mathbb{R}$, and $b,d\in\mathbb{C}$. Define $\psi=\angle (2ad+b)$ and $r=|\frac{2ad+b}{2a}|$, then the optimal solution $(x^*,y^*)$ for {\black problem \eqref{eq:simp_QCQP_y}} is
		\begin{equation}
		\begin{aligned}
		y^*&=\min\{ -\frac{\beta}{2\alpha},\frac{2a\delta(c-r)-\beta}{2a\delta^2+2\alpha},c/\delta  \},\\
		x^*&=d-ce^{j\psi}+e^{j\psi}\max\{  c-r, \delta y^* \}.
		\end{aligned}
		\end{equation}
	\end{lem}
	\begin{proof}
		See Appendix A in supplementary materials. 
	\end{proof}
	\subsection{Solution for Subproblem \eqref{eq:admm_w}}
	Subproblem \eqref{eq:admm_w} with respect to $\mathbf{w}$ and $y$ is a convex SOCP, given as
	\begin{equation}\label{eq:opt_sub_wy}
	\begin{aligned}
	\min_{\mathbf{w}}\quad&\mathbf{w}^H\mathbf{A}\mathbf{w}+\textrm{Re}\{ \mathbf{b}^H\mathbf{w}\} + \alpha y^2 + \beta y\\
	\textrm{s.t.} \quad& \| \mathbf{w} \| \leq y,
	\end{aligned}
	\end{equation}
	where $\mathbf{A}\succ 0$, $\mathbf{b}$, $\alpha$, and $\beta$ are 
	\begin{equation}
	\begin{aligned}
	&\mathbf{A} = \mathbf{R}+\frac{\rho}{2}(\sum_{\theta\in\Theta}{\mathbf{a}}_{\theta}{\mathbf{a}}_{\theta}^H+\sum_k\sum_{\phi\in\Phi_k}{\mathbf{a}}_{\phi}{\mathbf{a}}_{\phi}^H),\label{eq:cal_A}
	\end{aligned}
	\end{equation}
	\begin{equation}
	\begin{aligned}
	&\mathbf{b}=-\frac{1}{2}[\sum_{\theta\in\Theta}(\lambda_\theta^H{\mathbf{a}}_{\theta}-\rho z_\theta^H{\mathbf{a}}_{\theta})+\sum_k\sum_{\phi\in\Phi_k}(\lambda_\phi^H{\mathbf{a}}_{\phi}-\rho z_\phi^H{\mathbf{a}}_{\phi})],\notag
	\end{aligned}
	\end{equation}
	\begin{equation}
	\begin{aligned}
	&\alpha = (|\Theta| + \sum_k|\Phi_k|)\frac{\rho}{2},\beta = \sum_{\theta\in\Theta}\eta_\theta-\rho y_\theta+\sum_k\sum_{\phi\in\Phi_k}\eta_\phi-\rho y_\phi.\notag
	\end{aligned}
	\end{equation}
	Problem \eqref{eq:opt_sub_wy} has a specific structure, whose optimal solution can be obtained in a closed-form based on bisection search. Details are provided in Lemma \ref{lem:bisec_wy}. 
	\begin{lem}\label{lem:bisec_wy}
		Let eigenvalue decomposition of $\mathbf{A}$ be $\mathbf{U}\bm{\Lambda}\mathbf{U}^H$, then the optimal $(\mathbf{w}^*,y^*)$ for problem \eqref{eq:opt_sub_wy} is 
		\begin{align*}
		\left\lbrace
		\begin{aligned}
		&y^*=0,\mathbf{w}^*=\mathbf{0},\  \textrm{if }\beta\geq \| \mathbf{b}\|,&\\
		&y^*=-\frac{\beta}{2\alpha},\mathbf{w}^*=\mathbf{w}(y^*),\  \textrm{if }\beta< \|\mathbf{b}\|,\| \mathbf{A}^{-1} \mathbf{b}\|\leq -\frac{\beta}{\alpha},&\\
		&f(y^*)=1, \mathbf{w}^*=\mathbf{w}(y^*),\ \textrm{otherwise},
		\end{aligned}
		\right.
		\end{align*}
		where $\mathbf{w}(y)=-\mathbf{U}^H\left[  2\bm{\Lambda}+(2\alpha  + \frac{\beta}{y})\mathbf{I} \right]^{-1}\mathbf{U}\mathbf{b},$
		and $f(y)$ is a strictly monotonic decreasing function given in (45) (in supplementary materials). Further, the unique solution for $f(y)=1$ can be obtained by bisection search for $y\in[\max\{0,-\frac{\beta}{2\alpha}\},\| \mathbf{A}^{-1}\mathbf{{b}}\|/2]$.
	\end{lem}
	\begin{proof}
		See Appendix B in supplementary materials. 
	\end{proof}
	The computational cost for solving problem \eqref{eq:opt_sub_wy} per iteration includes: computing $\mathbf{b}$ with a complexity $\mathcal{O}(M(|\Theta|+\sum_k|\Phi_k|))$, computing $\mathbf{w}^*$ and $y^*$ with a complexity of $\mathcal{O}(M^2)$, and bisection search for $y$ with a complexity of $\mathcal{O}(\log_2\frac{1}{\Delta})$, where $\Delta$ is the numerical precision. In addition, \textcolor{black}{a one-time computational cost for finding the eigenvalues and eigenvectors of $\mathbf{A}$ with a complexity of $\mathcal{O}(M^{2.376})$} is required, since $\mathbf{A}$ does not change through ADMM iterations. In total, the computation complexity per iteration is $\mathcal{O}(M^2+M(|\Theta|+\sum_k|\Phi_k|)+\log_2\frac{1}{\Delta})$.
	\subsection{Solution for Subproblem \eqref{eq:admm_theta}}
	Subproblem \eqref{eq:admm_theta} is separable over each $(y_\theta,z_\theta),\theta\in\Theta$. Thus each optimal $(y_\theta^*,z_\theta^*),\theta\in\Theta$ can be obtained by individually solving the following problem,
	\begin{equation}\label{eq:sub_delta_theta}
	\begin{aligned}
	\min_{y_\theta,z_\theta}\ & a_\theta |z_\theta|^2 + \textrm{Re}\{ b_\theta^H z_\theta\} + \alpha_\theta y_\theta^2 + \beta_\theta y_\theta\\
	\textrm{s.t.}\ &|z_\theta-1| + \delta y_\theta \leq c_\theta,
	\end{aligned}
	\end{equation}
	where $a_\theta$, $b_\theta$, $\alpha_\theta$, and $\beta_\theta$ are 
	\begin{align*}
	&a_\theta = \alpha_\theta=\frac{\rho}{2},\ b_\theta = -\lambda_\theta-\rho\mathbf{w}^H\mathbf{a}_\theta,\ \beta_\theta = -\eta_\theta-\rho y.
	\end{align*}
	Problem \eqref{eq:sub_delta_theta} is a complex SOCP with respect to $(y_\theta,z_\theta)$. It has the same mathematical form as the SOCP studied in Lemma \ref{lem:simp_QCQP_y}. Specify $a=a_\theta$, $b=b_\theta$, $c=c_\theta$, $\alpha=\alpha_\theta$, $\beta=\beta_\theta$, $\psi_\theta=\angle (2a_\theta + b_\theta)$, and $d=1$, we can obtain the closed-form solution for $(y_\theta^*,z_\theta^*)$, given as 
	\begin{equation}\label{eq:yz_theta}
	\begin{aligned}
	y_\theta^* &= \min\{ -\frac{\beta_\theta}{2\alpha_\theta},\frac{\delta(2a_\theta c_\theta- |2a_\theta+ b_\theta|)-\beta_\theta}{2a_\theta\delta^2+2\alpha_\theta},c_\theta/\delta  \},\\
	z_\theta^* &= 1-e^{j \psi_\theta}\min\{  \frac{|2a_\theta+b_\theta|}{2a_\theta}, c_\theta-\delta y_\theta^*  \}
	\end{aligned}
	\end{equation}
	The effort to solve for $(y_\theta,z_\theta)$ involves calculating the inner product $\mathbf{w}^H\mathbf{a}_{\theta}$ which has a complexity of  $\mathcal{O}(M)$ and updating $(y_\theta,z_\theta)$ which has a complexity of $\mathcal{O}(1)$. The overall complexity, for all $\theta\in\Theta$, is $\mathcal{O}((M+1)|\Theta|)$ per iteration.
	
	\subsection{Solution for Subproblem \eqref{eq:admm_phi}} 
	Subproblem \eqref{eq:admm_phi} with respect to $\{y_\phi, z_\phi \}$ and $t$ is equivalent to
	\begin{equation}\label{eq:icmv_sub_epsilon_reformula}
	\begin{aligned}
	\min_{t,\{ y_\phi,z_\phi \}}\ &\mu t
	+\sum_k\sum_{\phi\in\Phi_k}a_\phi |z_\phi|^2 + \textrm{Re}\{ b_\phi^H z_\phi\} + \alpha_\phi y_\phi^2 + \beta_\phi y_\phi\\
	\textrm{s.t.}\quad &|z_\phi|+\delta y_\phi \leq t c_{\phi}/\gamma_k,\  \forall\phi\in\Phi_k,\ k=1,\ldots,K,
	\end{aligned}
	\end{equation}
	where $a_\phi$, $b_\phi$, $\alpha_\phi$, and $\beta_\phi$ are 
	\begin{align*}
	&a_\phi = \alpha_\phi=\frac{\rho}{2},\ b_\phi = -\lambda_\phi-\rho\mathbf{w}^H\mathbf{a}_\phi,\ \beta_\phi = -\eta_\phi-\rho y.
	\end{align*}
	Problem \eqref{eq:icmv_sub_epsilon_reformula} is a convex SOCP with respect to $t$ and all $(y_\phi,z_\phi  )$. By exploiting its special problem structure, we can also solve it in closed-form. In what follows, we will present the way to exploit problem structure to obtain the closed-form solution. For any fixed $\bar{t}\in \mathbb{R}$, problem \eqref{eq:icmv_sub_epsilon_reformula} is separable over each $(y_\phi,z_\phi),\forall \phi \in \Phi_k,\forall k$, given as 
	\begin{equation}\label{eq:yz_phi}
	\begin{aligned}
	\min_{ y_\phi,z_\phi }\ &a_\phi |z_\phi|^2 + \textrm{Re}\{ b_\phi^H z_\phi\} + \alpha_\phi y_\phi^2 + \beta_\phi y_\phi\\
	\textrm{s.t.}\quad &|z_\phi|+\delta y_\phi \leq \bar{t} c_{\phi}/\gamma_k.
	\end{aligned}
	\end{equation}
	According Lemma \ref{lem:simp_QCQP_y}, the closed-form for optimal $(y^*_\phi,z_\phi^*)$ can be obtained by specifying $b=b_\phi$, $c=\bar{t}c_\phi/\gamma_k$, $\alpha=\alpha_\phi$, $\beta=\beta_\phi$, $\psi_\phi=\angle b_\phi$, and $d=0$, $\forall \phi\in\Phi_k,\forall k$, given as 
	\begin{equation}\label{eq:yz_phi_solu}
	\begin{aligned}
	y_\phi^* &= \min\{ -\frac{\beta_\phi}{2\alpha_\phi},\frac{\delta(2a_\phi \bar{t}c_\phi/\gamma_k- | b_\phi|)-\beta_\phi}{2a_\phi\delta^2+2\alpha_\phi},\frac{\bar{t}c_\phi}{\gamma_k\delta}\},\\
	z_\phi^* &= -e^{j \psi_\phi}\min\{  \frac{|b_\phi|}{2a_\phi}, \bar{t}c_\phi/\gamma_k-\delta y_\phi^*  \}.
	\end{aligned}
	\end{equation}
	Notice $(y_\phi^*,z_\phi^*)$ is a function of $\bar{t}$, which implies problem \eqref{eq:yz_phi} can be reduced to an optimization problem with respect to a single variable $t$. Hence, {\black the key becomes} how to find optimal $t^*$ for the reduced problem. In Lemma \ref{lem:bisec_SOCP_t}, we first give the equivalent optimization problem with respect to $t$, and then in Proposition \ref{prop:solu_t}, \textcolor{black}{we show that the optimal $t^*$ can be obtained by sorting a set of real numbers.}
	\begin{lem}\label{lem:bisec_SOCP_t}
		The optimal $t^*$ for problem \eqref{eq:icmv_sub_epsilon_reformula} can be obtained by solving the following unconstrained convex problem
		\begin{equation}\label{eq:opt_single_t}
		\min_t\ \mu t + f(t),
		\end{equation}
		where $f(t)\triangleq \sum_k\sum_{\phi\in\Phi_k} f_\phi(t) $  and $f_\phi(t)$ is a convex function defined as 
		\begin{equation}
		f_\phi(t)=\left\lbrace
		\begin{aligned}
		&a_{\phi,1} t^2 + b_{\phi,1}t,\  t\leq\bar{t}_{\phi,1},&\\
		&a_{\phi,2} t^2 + b_{\phi,2}t + c_{\phi,2},\  \bar{t}_{\phi,1}<t\leq\bar{t}_{\phi,2},&\\
		&c_{\phi,3},\ \bar{t}_{\phi,2}<t,&
		\end{aligned}
		\right.
		\end{equation}
		and  
		\begin{equation}\label{eq:opt_t_coeff}
		\begin{aligned}
		&a_{\phi,1}=\frac{\alpha_\phi c_\phi^2}{\gamma_k^2\delta^2},\ b_{\phi,1}=\frac{\beta_\phi c_\phi}{\gamma_k\delta},\ c_{\phi,3}=-\frac{b_\phi^2}{4a_\phi}-\frac{\beta_\phi^2}{4\alpha_\phi},\\
		&a_{\phi,2}=\frac{\alpha_\phi a_\phi c_\phi^2}{\gamma_k^2(\alpha_\phi+a_\phi \delta^2)},\ b_{\phi,2} = \frac{\delta a_\phi c_\phi \beta_\phi - \alpha_\phi c_\phi |b_\phi| }{\gamma_k(\alpha_\phi+a_\phi \delta^2)},\\
		&c_{\phi,2}=-\frac{(\delta|b_\phi| + \beta_\phi)^2}{4(\alpha_\phi + a_\phi \delta^2)},\ \bar{t}_{\phi,1} = -\frac{\gamma_k(|b_\phi|\delta^2 + \beta_\phi\delta )}{2 \alpha_\phi c_\phi},\\
		&\bar{t}_{\phi,2} = -\frac{\gamma_k(\beta_\phi\delta - |b_\phi|\alpha_\phi/a_\phi  )}{2 \alpha_\phi c_\phi}.
		\end{aligned}
		\end{equation}
		Further, $f(t)$ is a smooth convex function.
	\end{lem}
	\begin{proof}
		Substituting \eqref{eq:yz_phi_solu} into \eqref{eq:yz_phi} completes the proof.
	\end{proof}
	
	\begin{prop}\label{prop:solu_t}
		Let $\{ \tilde{t}_\ell\}$  be an increasing sequence consisted by all $\{ \bar{t}_{\phi,1} \}$ and $\{ \bar{t}_{\phi,2} \}$, denote $\partial f({t})$ as the derivative of $f(t)$ in Lemma \ref{lem:bisec_SOCP_t}. If $\partial f(\tilde{t}_1) > -\mu$, then the optimal $t^*$ for problem \eqref{eq:opt_single_t} is 
		$$t^*=-\frac{\sum_k\sum_{\phi\in\Phi_k} b_{\phi,1}+\mu}{2\sum_k\sum_{\phi\in\Phi_k} a_{\phi,1}}.$$ 
		Otherwise there must exist $\tilde{t}_\ell$ satisfying $\partial f(\tilde{t}_\ell) \leq -\mu$ and $\partial f(\tilde{t}_{\ell+1}) \geq -\mu$, and the optimal $t^*$ for problem \eqref{eq:opt_single_t} is 
		\begin{equation}
		t^*=-\frac{\sum_{\phi\in\Omega_1} b_{\phi,1}+\sum_{\phi\in\Omega_2}b_{\phi,2}+\mu}{2(\sum_{\phi\in\Omega_1} a_{\phi,1}+\sum_{\phi\in\Omega_2}a_{\phi,2})},
		\end{equation}
		where sets $\Omega_1$ and $\Omega_2$ are defined as 
		$$
		\Omega_1=\{  \phi | \bar{t}_{\phi,1}\geq \tilde{t}_{\ell+1} \},\ \Omega_2 = \{   \phi |  \bar{t}_{\phi,1}\leq \tilde{t}_{\ell}, \bar{t}_{\phi,2}\geq \tilde{t}_{\ell+1}  \}.
		$$
	\end{prop}
	\begin{proof}
		See Appendix C in supplementary materials. 
	\end{proof}

	By \eqref{eq:yz_phi_solu} and Proposition \ref{prop:solu_t}, the computational {\black cost} for solving problem~\eqref{eq:icmv_sub_epsilon_reformula} includes: computing the inner products $\{\mathbf{w}^H\mathbf{a}_{\phi}\}$ with a complexity of $\mathcal{O}(M\sum_k|\Phi_k|)$, computing the coefficients in \eqref{eq:opt_t_coeff} with a complexity of $\mathcal{O}(\sum_k|\Phi_k|)$, sorting $2\sum_k|\Phi_k|$ points for finding $t^*$ with a complexity of $\mathcal{O}((2\sum_k|\Phi_k|)\log_2(2\sum_k|\Phi_k|))$, and extracting $\{ y_\phi^*,z_\phi^* \}$ by~\eqref{eq:yz_phi_solu} with a complexity of $\mathcal{O}(\sum_k|\Phi_k|)$. Totally, the complexity is $\mathcal{O}(M\sum_k|\Phi_k|+(2\sum_k|\Phi_k|)\log_2(2\sum_k|\Phi_k|)+\sum_k|\Phi_k|)$.
	\subsection{Update $\{\lambda_\theta\}$ and $\{\lambda_\phi\}$} 
	See \eqref{eq:admm_lamtheta}-\eqref{eq:admm_etaphi}. Since the inner products $\mathbf{w}^H\mathbf{a}_\theta$ and $\mathbf{w}^H\mathbf{a}_\phi$ are calculated when solving subproblems \eqref{eq:admm_theta} and \eqref{eq:admm_phi}. The complexity for updating $\{\eta_{\phi}, \lambda_\theta\}$ and $\{\eta_{\phi},\lambda_\phi\}$ is only $\mathcal{O}(|\Theta|+\sum_k|\Phi_k|)$.
	
	\subsection{Proposed ADMM Algorithm}\label{subsec:admm}
	We summarize the proposed ADMM algorithm for the reformulated P-ICMV formulation \eqref{eq:admm_picmv} in Algorithm \ref{Alg}. The total computation complexity of the proposed ADMM algorithm per iteration is $\mathcal{O}(M^2+(M+1)(C_1+C_2)+\log_2\frac{1}{\Delta}+C_2(\log_2(C_2)+1))$ where $C_1=|\Theta|$ and $C_2=\sum_k|\Phi_k|$, and the complexity for {\black each subproblem} is summarized in Table~\ref{tab:complexity}. It should be noticed that subproblems \eqref{eq:admm_theta} and \eqref{eq:admm_phi} (with given $t$) are separable across $\theta\in\Theta$ and $\phi\in\Phi_k,\forall k$ respectively, parallel implementation for solving these subproblems can further improve the implementation efficiency.
	\begin{table}
	\centering
	\caption{Computation complexity of each subproblem.}
	\label{tab:complexity}
    \begin{tabular}{c c}
    \toprule
        & Complexity \\\midrule
       Subproblem~\eqref{eq:admm_w} & $\mathcal{O}(M^2+M(C_1+C_2)+\log_2\frac{1}{\Delta})$ \\
       Subproblem~\eqref{eq:admm_phi} & $\mathcal{O}((M+1)|\Theta|)$\\
       Subproblem~\eqref{eq:admm_lamtheta} & $\mathcal{O}(MC_2+2C_2\log_2(2C_2)+C_2)$\\
       Subproblem~\eqref{eq:admm_lamtheta}-\eqref{eq:admm_etaphi} &  $\mathcal{O}(C_1+C_2)$\\ \bottomrule
    \end{tabular}
    \end{table}
	
	\renewcommand{\algorithmicrequire}{\textbf{Input:}} 
	\renewcommand{\algorithmicensure}{\textbf{Output:}} 
	\begin{algorithm}[ht]
		\caption{Proposed ADMM Algorithm for Problem \eqref{eq:admm_picmv}}
		\label{Alg}
		\begin{algorithmic}[1]
			\Require
			$\mathbf{R}$, $\{ \mathbf{a}_\theta  \},\{ \mathbf{a}_\phi \}$, $\{  c_\theta\},\{  c_\phi \}$, $\{ \gamma_k \}$, $\rho$, $\mu,\delta$
			\State Compute $\mathbf{A}$ by \eqref{eq:cal_A};
			\For{$r=0,1,\ldots,$ until meet some convergence criteria}
			\State Update $\mathbf{w}^{r+1}$ and $y^{r+1}$by Lemma \ref{lem:bisec_wy};
			\State Compute $\{ \mathbf{a}_\theta^H\mathbf{w}^{r+1}  \},\{ \mathbf{a}_\phi^H\mathbf{w}^{r+1} \}$
			\State Update $\{ z_\theta^{r+1}\}$ and $\{ y_\theta^{r+1} \} $by \eqref{eq:yz_theta};
			\State Compute $t^{r+1}$ by Proposition \ref{prop:solu_t};
			\State Update $\{ z_\phi^{r+1} \}$ and $\{  y_\phi^{r+1} \}$ by \eqref{eq:yz_phi_solu};
			\State Update $\{ \eta_{\theta}^{r+1},\lambda_\theta^{r+1} \}$, $\{  \eta_{\phi}^{r+1},\lambda_{\phi}^{r+1} \}$ by \eqref{eq:admm_lamtheta}-\eqref{eq:admm_etaphi};
			\EndFor
			\Ensure
			The P-ICMV beamformer $\mathbf{w}^*$.
		\end{algorithmic}
	\end{algorithm}
	
	\begin{rmk}[ADMM for Beamforming:]
	ADMM is a primal-dual algorithm which has recently been used in many array beamforming applications, e.g.,~\cite{fan2019robust,liang2018sparse,cheng2017constant,yu2020quadratic,gemechu2019beampattern,feng2020phased}, just to mention a few. The flexibility of variable splitting technique makes it adaptable to different beamforming problem structure which results specific primal subproblems and subsequently different solution techniques are developed to different primal subproblems. The proposed ADMM algorithm involves solving special SOCPs (see~\eqref{eq:opt_sub_wy}, \eqref{eq:sub_delta_theta}, and \eqref{eq:icmv_sub_epsilon_reformula}) in the primal steps and the developed technique is different from previous techniques which are used to solve other special subproblems, e.g., quadratically constrained quadratic programs~\cite{fan2019robust,liang2018sparse,cheng2017constant,yu2020quadratic,gemechu2019beampattern}, complex $\ell_p$ norm problem~\cite{fan2019robust,liang2018sparse}, projection problem over discrete set~\cite{yu2020quadratic,feng2020phased}.
	\end{rmk}

	\section{Numerical Simulation}\label{sec:simul}
	{\black Three different beamforming applications are considered in this section to demonstrate the effectiveness of the proposed P-ICMV beamformer.} The first is robust adaptive beamforming with antenna array, which is widely used in radar and wireless communication systems. The major challenge for this application is how to robustly handle various errors, including DoA error, SV mismatch error, and covariance estimation error. The robustness of the P-ICMV beamformer against these errors is demonstrated. The second application is speech enhancement with microphone array in hearing aids, where the major difficulty is handling multiple interfering speakers with limited number of microphones. The efficiency of suppressing multiple interferences with few number of microphones is verified. The third application is beam pattern synthesis for a large size array. A specified beam pattern for a $30\times 30$ antenna array is synthesized, {\black where tens of thousands of inequality constraints are enforced to achieve a desired beam pattern. The computational efficiency of the proposed ADMM algorithm is demonstrated.} For all simulations, the ADMM algorithm stops when both the primal-dual feasibility gap and residual~\cite{boyd2011distributed} less than $10^{-5}$ or the number of iterations exceed $10^3$.
	
	\subsection{Robust Adaptive Beamforming with Antenna Array}\label{subsec:rbst}
	We consider a uniform linear antenna array of $M=20$ elements with {\black half-wave} length spacing. The desired signal is at $\theta_0=-5^\circ$, and three interfering signals are at $\theta_1=-60^\circ$, $\theta_2=-20^\circ$ and $\theta_3=45^\circ$ respectively. For each interfering signal, the interference-to-noise ratio (INR) is fixed at $30$ dB. {\black The covariance matrix $\mathbf{R}$ is estimated by averaging a finite number of training snapshots and $10^{-6}\times\lambda_{\rm max}(\mathbf{R})$ is added to its diagonal to stablize the numerical computation. In all simulations, the desired signal is always present in the snapshots and both the DoA error and the array calibration error are considered.} {\black Specifically, the estimated DoA $ \hat{\theta}_k$ for each source $k$ is uniformly drawn from $\mathcal{U}({\theta}_k-2^\circ,{\theta}_k+2^\circ)$ and the largest DoA estimation error is presumed to be $\Delta=4^\circ$. This prior konwledge will be used to specify different beamformer's parameter.} {\black The presumed SVs with $1^\circ$ separation over region $[-90^\circ,90^\circ]$ are assumed to b e available.} For each antenna element, its gain and phase errors are randomly generated from $\mathcal{N}(1, 0.02^2)$ and $\mathcal{N}(0, (0.01\pi)^2)$ respectively. The beamforming output SINR is used as performance metric, and output SINRs under different input SNR and number of snapshots conditions are studied. For each simulation condition, the output SINR is averaged by $100$ independent Monte Carlo runs. 
	
	Five representative robust beamformers in the literature are selected {\black in the} comparison. {\black The first} beamformer is the loading sample matrix inversion (LSMI) beamformer~\cite{Elnashar2006Diagonal} {\black and the loading factor set to be $10\lambda_{\textrm{min}}({\mathbf{R}})$.} The second beamformer is the worst-case beamformer~\cite{Vorobyov2003worst} and the robust parameter $\epsilon$ is set to be $3$. The third beamformer is the eigenspace beamformer \cite{Feldman1996projection} where the number of interferences is assumed to be exactly known. The fourth one is the reconstruction beamfomer~\cite{Gu2012Reconstruction} and the last is the SV estimation beamformer \cite{Khabbazibasmenj2012Estimation}. {\black The desired signal region for these two beamformers is specified as $\hat{\theta}_0+[-\Delta,\Delta]$. For the P-ICMV beamformer, the discrete angle set $\Theta$ is optimistically specified over region $\hat{\theta}_0+[-\Delta/2,\Delta/2]$, i.e., $\Theta=\hat{\theta}_0+\{ -2^\circ,-1^\circ,0^\circ,1^\circ,2^\circ\}$, and the corresponding parameters $c_\theta,\forall \theta\in\Theta,$ are set to be proportional to angle errors, i.e., $c_\Theta=\{6,4,2,4,6  \}\times 10^{-1}$.} In such a case, the sufficient condition in Proposition \ref{prop:suffi_fea} for $\delta$ is about $\delta\leq 0.24$. {\black The angle set $\Phi_k$ for each interfering signal $k$ is specified  over region $\hat{\theta}_k+[-\Delta,\Delta]$. To further capture the environment information for adaptive interference suppression, we use ${\mathbf{R}}$ and the presumed SVs $\{ \mathbf{a}_\phi\}$ to specify parameters $\{ c_\phi \}$ and $\{  \gamma_k \}$. Precisely, for each $\Phi_k$, $c_\phi$ is set as $c_\phi=\hat{\sigma}_\phi^{-1}/\max_{\phi\in\Phi_k}\{ \hat{\sigma}_\phi^{-1}   \}$, where $\hat{\sigma}_\phi^2=1/\mathbf{a}_\phi^H{\mathbf{R}}^{-1}\mathbf{a}_\phi$ is the estimation of the so-called Capon spectrum~\cite{capon1969high}. For $\{\gamma_k \}$, we specify them as $\gamma_k=\beta_k/\max_{k^\prime}\{  \beta_{k^\prime}  \}$, where $\beta_k=\sum_{\phi\in\Phi_k}\hat{\sigma}_\phi^2$.}
	
	In Figs.~\ref{fig:simul_ant_snr} and~\ref{fig:simul_ant_snap}, the output SINRs of the P-ICMV beamformer are compared with other five beamformers. Parameter $\mu$ and $\delta$ of the P-ICMV beamformer are set as $\mu=10\lambda_{\textrm{max}}({\mathbf{R}})$ and  $\delta=10^{-2}$ respectively, and the parameter $\rho$ in ADMM algorithm is $\rho=10\mu$.  Fig.~\ref{fig:simul_ant_snr} plots the output SINRs under different SNR condition with $40$ snapshots. It can be observed that the P-ICMV beamformer achieves a stable output SINR and there is about $3$dB performance degradation from the optimal SINR among all SNR conditions. The performance of the reconstruction beamformer is similar (less than 1 dB difference) to the P-ICMV beamformer since {\black both of them} exploit the presumed SVs of interferences in the beamformer design. {\black The other four beamformers do not utilize the presumed SVs of interferences, they suffer more performance degradation as SNR increases.} In Fig.~\ref{fig:simul_ant_snap}, the output SINRs under different number of snapshots are compared while the input SNR is fixed at $15$dB. The reconstructioin and P-ICMV beamformers do not show obvious performance degradation as the number of snapshots decreases while the other beamformers do. {\black Since P-ICMV beamformer enforces constraints to suppress interferences, it performs better than the reconstructioin beamformer in the case with only 2 snapshots.} 
	\begin{figure}[h]
		\centering
		\includegraphics [width=0.6\linewidth]{ 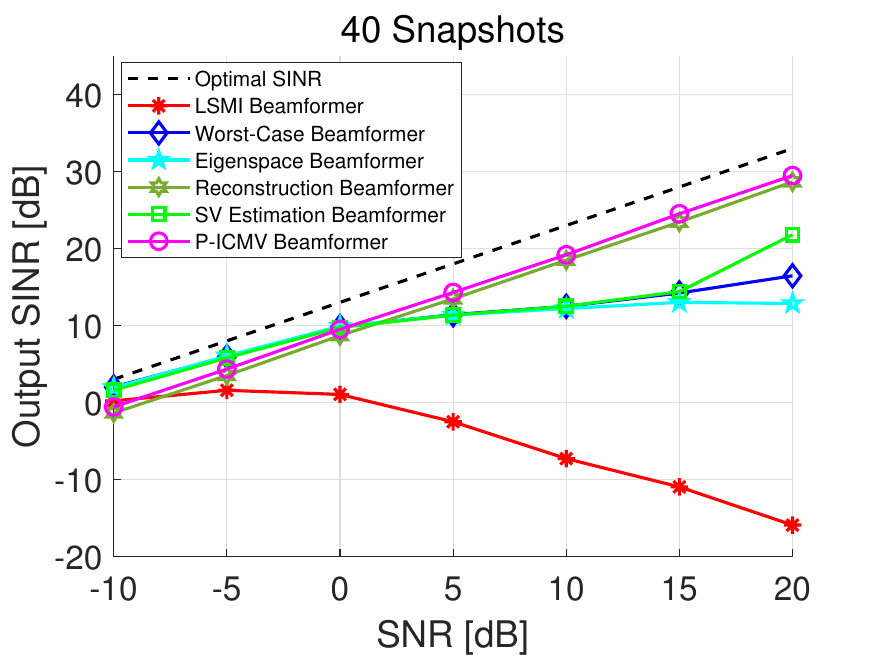}
		\caption{Output SINRs with different SNRs.}
		\label{fig:simul_ant_snr}
	\end{figure}
	\begin{figure}[h]
		\centering
		\includegraphics [width=0.6\linewidth]{ 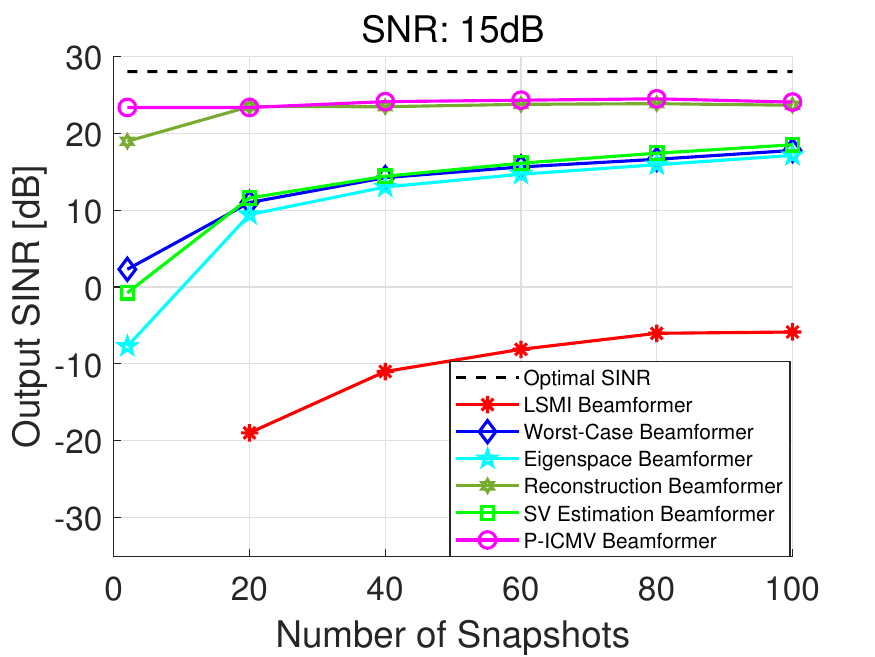}
		\caption{Output SINRs with different numbers of snapshots.}
		\label{fig:simul_ant_snap}
	\end{figure}
	
	The impact of the two trade-off parameters $\mu$ and $\delta$ is studied in Fig.~\ref{fig:simul_ant_nu}. As $\delta$ increases, the output SINR slightly increases due to the true spatial response for the target signal direction is bounded with more allowable SV mismatch. However, as $\delta$ continues to increase, more DoFs are utilized to guarantee the bounded true spatial responses, and hence the output SINRs deteriorate. 
	\vspace{-0.5em}
	\begin{figure}[h]
		\centering
		\includegraphics [width=0.6\linewidth]{ 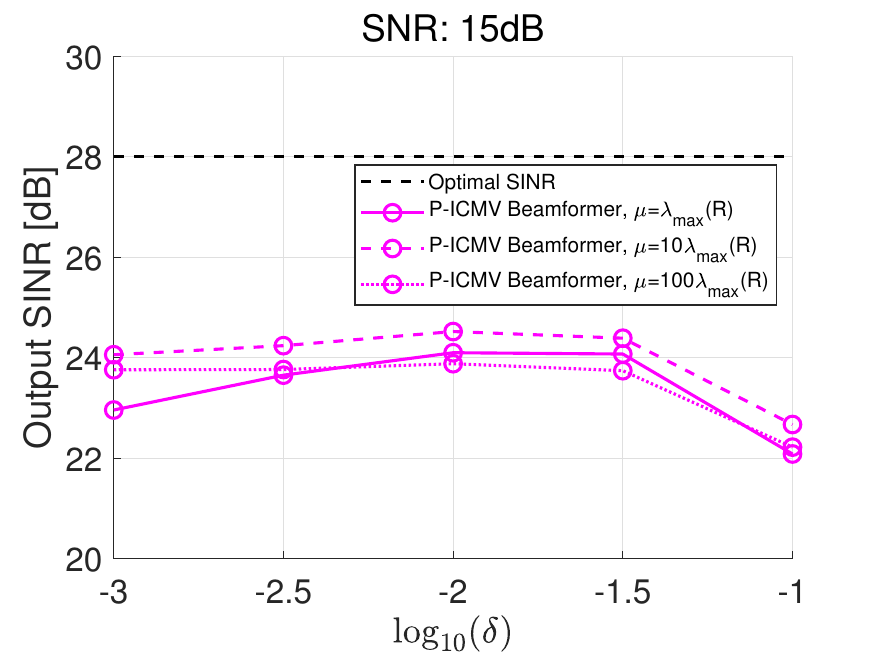}
		\caption{Output SINRs with different values of $\mu$ and $\delta$.}
		\label{fig:simul_ant_nu}
	\end{figure}
	
	\subsection{Speech Enhancement with Microphone Array}\label{subsec:mic}
	In this subsection, we consider a hearing aids application for speech enhancement with microphone array in a babble noise environment. Similar to the evaluation in \cite{Pu2017penalized,Xiao2018Evaluation}, two minimum variance-based beamformers: the LCMV and ICMV beamformers are selected {\black in the comparison.} A rectangular room of size $12.7$m$\times10$m with height $3.6$m is used for simulating the acoustic environment and the reverberation time is set to be $0.6$ second. The room impulse responses (RIRs) is generated by the so-called image method~\cite{allen1979image}.
	We specify the hearing aids wearer located at the center of the room, each hearing aid has 2 microphones with $7.5$mm spacing. The front microphone of the left is set as the reference microphone. The head shadow effect of the listener is also taken into account through using the measurement of the head-related relative transfer functions of the hearings aids on a mannequin. The simulated acoustic environment is illustrated in Fig.~\ref{fig:simul_sc}.
	\vspace{-0.5em}
	\begin{figure}[H]
		\centering
		\includegraphics [width=0.6\linewidth]{ 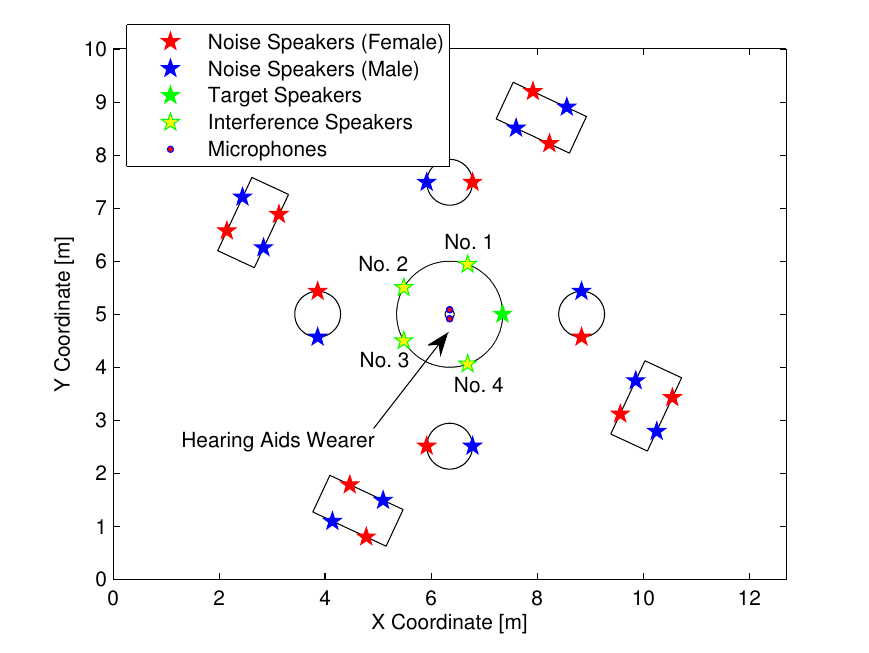}
		\caption{Simulated acoustic environment.}
		\label{fig:simul_sc}
	\end{figure}
	
	In the room, there are one target and four interference sources. The target and interference sources are represented by speakers $1$m away from the listener. The target is
	at $0^{\circ}$ and interferences are at $\pm 70^{\circ}$ and $\pm 150^{\circ}$ (corresponds to No.~$1$ to $4$ in Fig.~\ref{fig:simul_sc}). The background babble noise is simulated using $24$ speakers at different locations. All speakers and hearings aids microphones are in the same horizontal plane at a height of $1.2$m. All speakers' speech signals are taken from the TIMIT database~\cite{garofolo1993darpa}. For each speaker, there is $0.5$ second silence between each sentence and the speech lasts 25 seconds. In the beginning 3 seconds, only those babble speakers are active and such time segment is utilized for estimating the noise correlation matrix $\mathbf{R}$ by sample averaging. The input signal-to-noise (SNR) at the reference microphone is set at $10$dB and signal-to-interference ratio (SIR) for each interference is set at $-3$dB. The audio signals are sampled at $16$kHz and a $1024$-point FFT with $50\%$ overlap is used to transform the signals into the time-frequency domain. In the simulation, anechoic ATFs and DoAs of all sources are known. 
	
	Since there are in total 5 sources but only $4$ microphones, the LCMV and ICMV beamformers can only have constraints on 3 interferences besides the target. We use ``Setup $i$, $i=1,2,3,4$'' to denote the setting in which interferer $i$ is ignored and other remaining interferences are suppressed by the corresponding constraints. Specifically, define index set $T_i=\{1,2,3,4\}/i$, the LCMV beamformer formulation with setup $i$ is 
	\begin{align*}
	\min_{\mathbf{w}}\ \mathbf{w}^H\mathbf{R}\mathbf{w}\quad\textrm{s.t.}\ &\mathbf{w}^H\mathbf{a}_{\theta_0}=1,\\
	&\mathbf{w}^H\mathbf{a}_{\theta_k}=0,\ \forall k\in T_i,
	\end{align*}
	where $\theta_0$ is the DoA of the target source and $\theta_k$ is the DoA of interferer $k$. For the ICMV beamformer with setup $i$, the corresponding formulation is 
	\begin{align*}
	\min_{\mathbf{w}}\ \mathbf{w}^H\mathbf{R}\mathbf{w}\quad\textrm{s.t.}\ &|\mathbf{w}^H\mathbf{a}_{\theta_0}-1|^2\leq c_{\theta_0}^2,\\
	&|\mathbf{w}^H\mathbf{a}_{\theta_k}|^2\leq c_{\theta_k}^2,\ \forall k\in T_i,
	\end{align*}
	where parameter $\{c_{\theta_k}\}$ are all set to be $0.1$.
	As for the P-ICMV beamformer, the corresponding formulation becomes
	\begin{align*}
	\min_{\mathbf{w}}\ &\mathbf{w}^H\mathbf{R}\mathbf{w}+\mu\max_{k}\{ \gamma_k\epsilon_k \}\\
	\quad\textrm{s.t.}\ &|\mathbf{w}^H\mathbf{a}_{\theta_0}-1|+\delta\|\mathbf{w}\| \leq c_{\theta_0},\\
	&|\mathbf{w}^H\mathbf{a}_{\theta_k}|+\delta\| \mathbf{w} \|\leq \epsilon_kc_{\theta_k},\ \forall k.
	\end{align*}
	{\black To make a fair comparison with ICMV, parameters $\{c_{\theta_k}\}$ are set the same as those in ICMV. Other parameters are set as $\gamma_k=1,\forall k$, and $\mu=\lambda_{\textrm{max}}(\mathbf{R})$.} The sufficient condition for $\delta$ by Proposition~\ref{prop:suffi_fea} is $\delta\leq 0.63$ and we choose $\delta=0.01$. The impact of different $\mu$ and $\delta$ are studied in Figs.~\ref{fig:diff_mu} and~\ref{fig:diff_delta}. \textcolor{black}{Penalty parameter $\rho$ in ADMM algorithm is set to be $10\mu$ for all simulations.}
	\vspace{-0.5em}
	\begin{table}[h]
		\centering
		\caption{IW-SINRI and IW-SD[dB]}
		\begin{tabular}{@{ }l@{\quad\quad $\ $} c@{$\ $ }  c@{$\ $ } c@{$\ $ } c @{\quad\quad\quad}  c@{ $\ $} c@{$\ $ } c@{ $\ $} c@{ }}
			\toprule
			&\multicolumn{4}{c}{\textbf{IW-SINRI}}&\multicolumn{4}{c}{\textbf{IW-SD}}\\
			\textbf{Setup}& 1 &2 &  3&  4&1&2 &  3& 4\\
			\midrule
			\textbf{LCMV}&7.22  &-4.19 &  -0.11&  8.37   &0.83&2.11 & 2.02& \textbf{0.77}\\
			\textbf{ICMV}&8.64  &-0.88 & 2.82&  8.86   &1.18 &1.97 &  1.92& 1.12\\
			\textbf{P-ICMV}&\multicolumn{4}{c}{\textbf{9.35$\quad\quad\quad\quad$}} &\multicolumn{4}{c}{1.15} \\
			\bottomrule
		\end{tabular}
		\label{tab:limit_4mic}
	\end{table}
	
	The intelligibility-weighted SINR improvement (IW-SINRI) and intelligibility-weighted spectral distortion (IW-SD) are used as performance metrics~\cite{Spriet2005Robustness} and are compared in Table \ref{tab:limit_4mic}. 
	In all 4 setups, the P-ICMV beamformer achieves more interference and noise suppression than the LCMV and ICMV beamformers in terms of the IW-SINRI metric {\black and} all three beamformers have similar speech distortion in terms of IW-SD scores. 
	It can be observed that for the LCMV and ICMV beamformers in setups $1$ and $4$ when one front interference is ignored, the beamformer achieves reasonable interference suppression. However, in setups $2$ and $3$ when one interference in the rear one is ignored, the beamformer has a poor IW-SINRI.$\ $This can be explained from the individual interference suppression levels (ISL) and the corresponding snapshots of the beam patterns.
	Fig.~\ref{fig:simul_rmsr} plots the ISL for the 4 setups. The ISL is defined as $\textrm{ISL}\triangleq20\log_{10}\frac{r_{\textrm{in}}}{r_{\textrm{out}}}$, where $r_{\textrm{in}}$ is the root mean square (RMS) of signal at reference microphone and $r_{\textrm{out}}$ is its RMS at beamformer's output. We can see that the P-ICMV beamformers achieves around 10dB ISL for all interferences, while for the LCMV and ICMV beamformers,
	only the interferences with constraints are well suppressed. The ignored one is either slightly suppressed or even enhanced depending on the setups.
	\begin{figure}[h]
		\centering
		\includegraphics [width=0.7\linewidth]{ 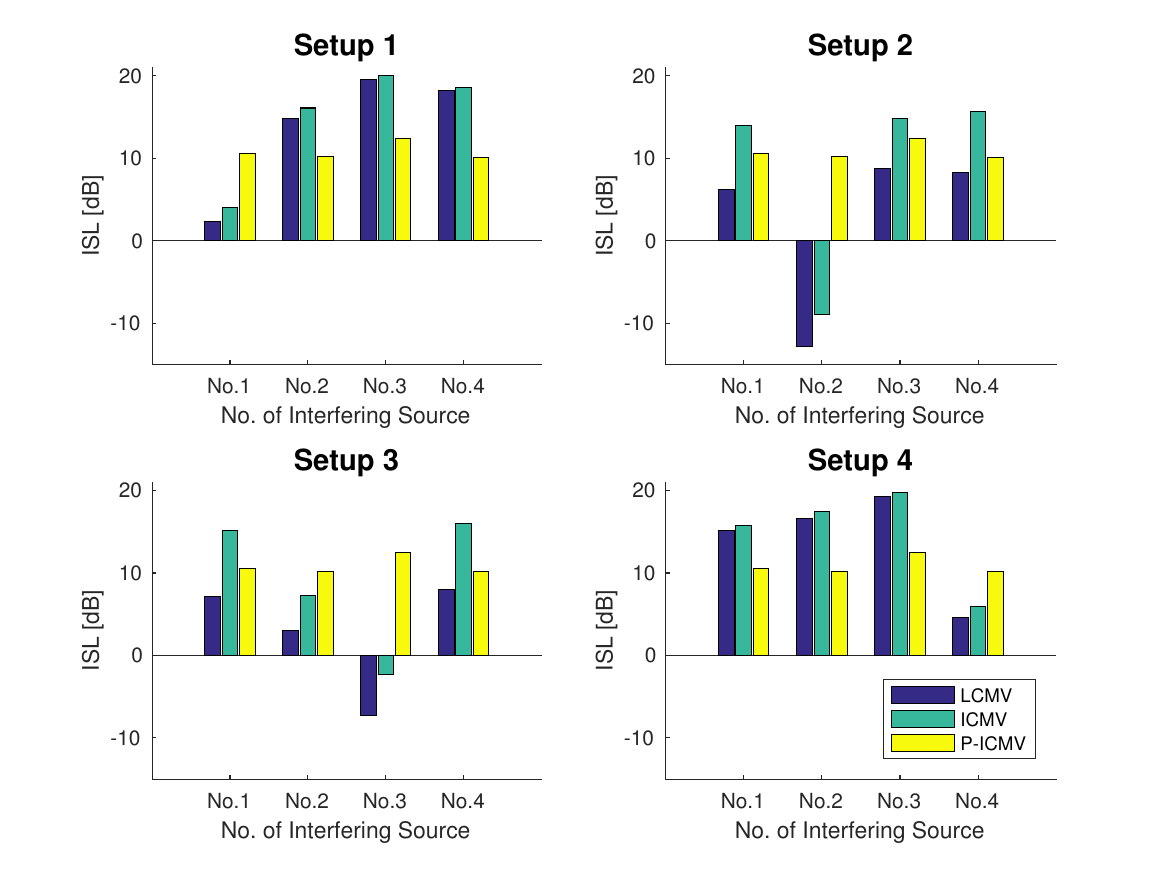}
		\caption{Individual interference suppression level.}
		\label{fig:simul_rmsr}
	\end{figure}

	\begin{figure}[h]
		\centering
		\includegraphics [width=0.4\linewidth]{ 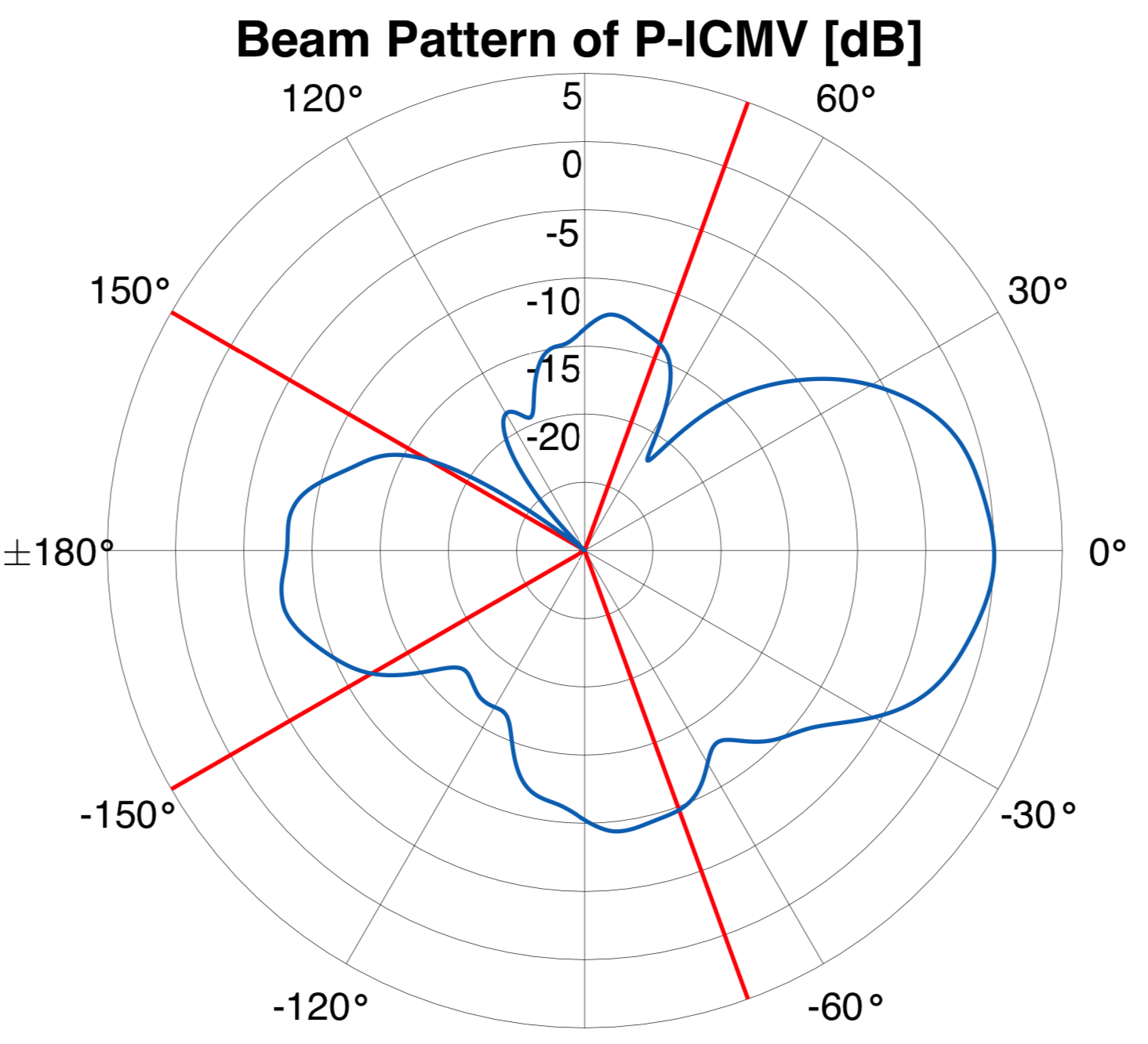}
		\caption{Beam patterns of P-ICMV at 1000 Hz.}
		\label{fig:simul_picmv}
	\end{figure}

	\begin{figure}[h]
		\subfigure[Setup 1]{\includegraphics [width=0.23\linewidth]{ 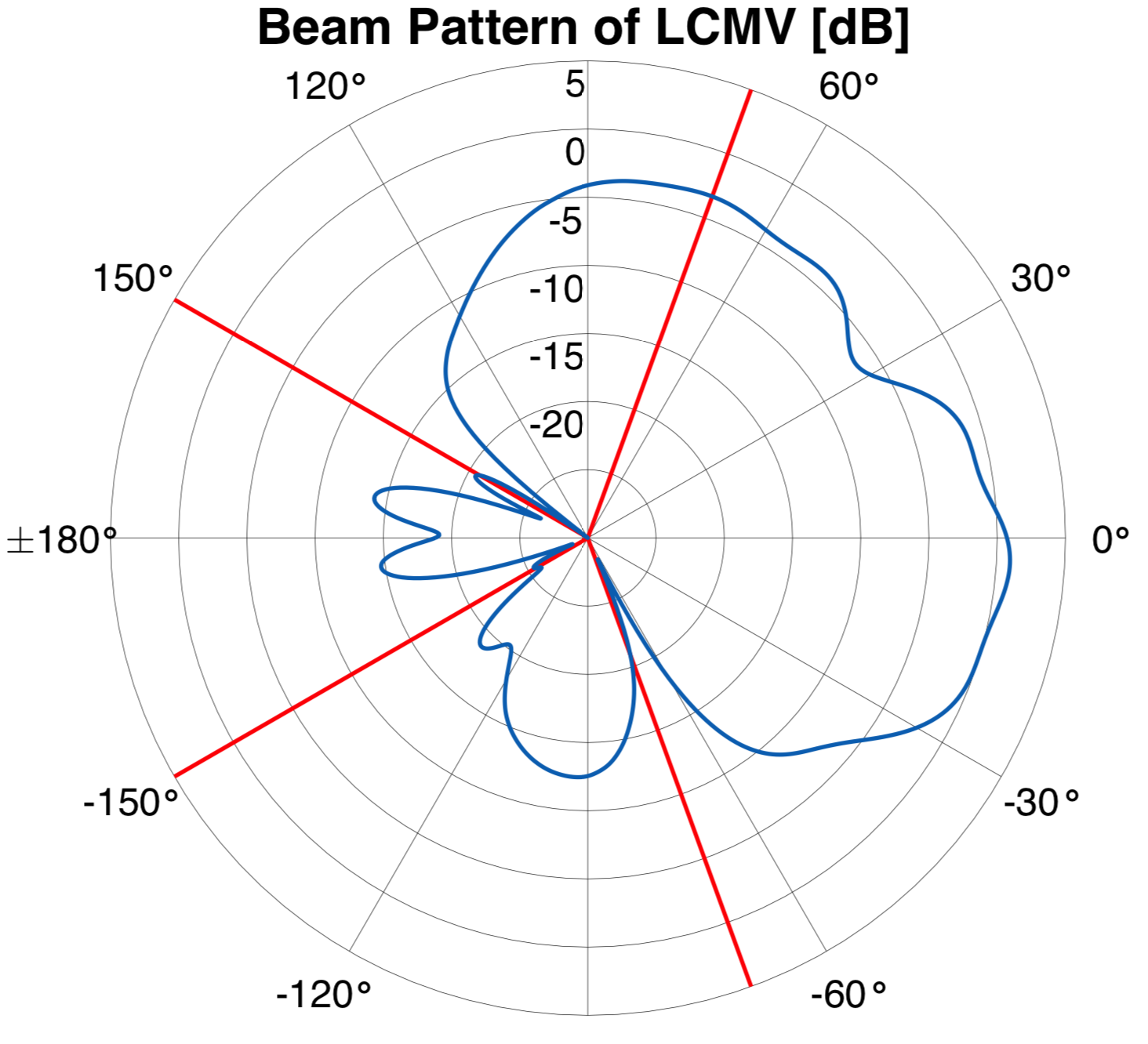}}
		\subfigure[Setup 2]{\includegraphics [width=0.23\linewidth]{ 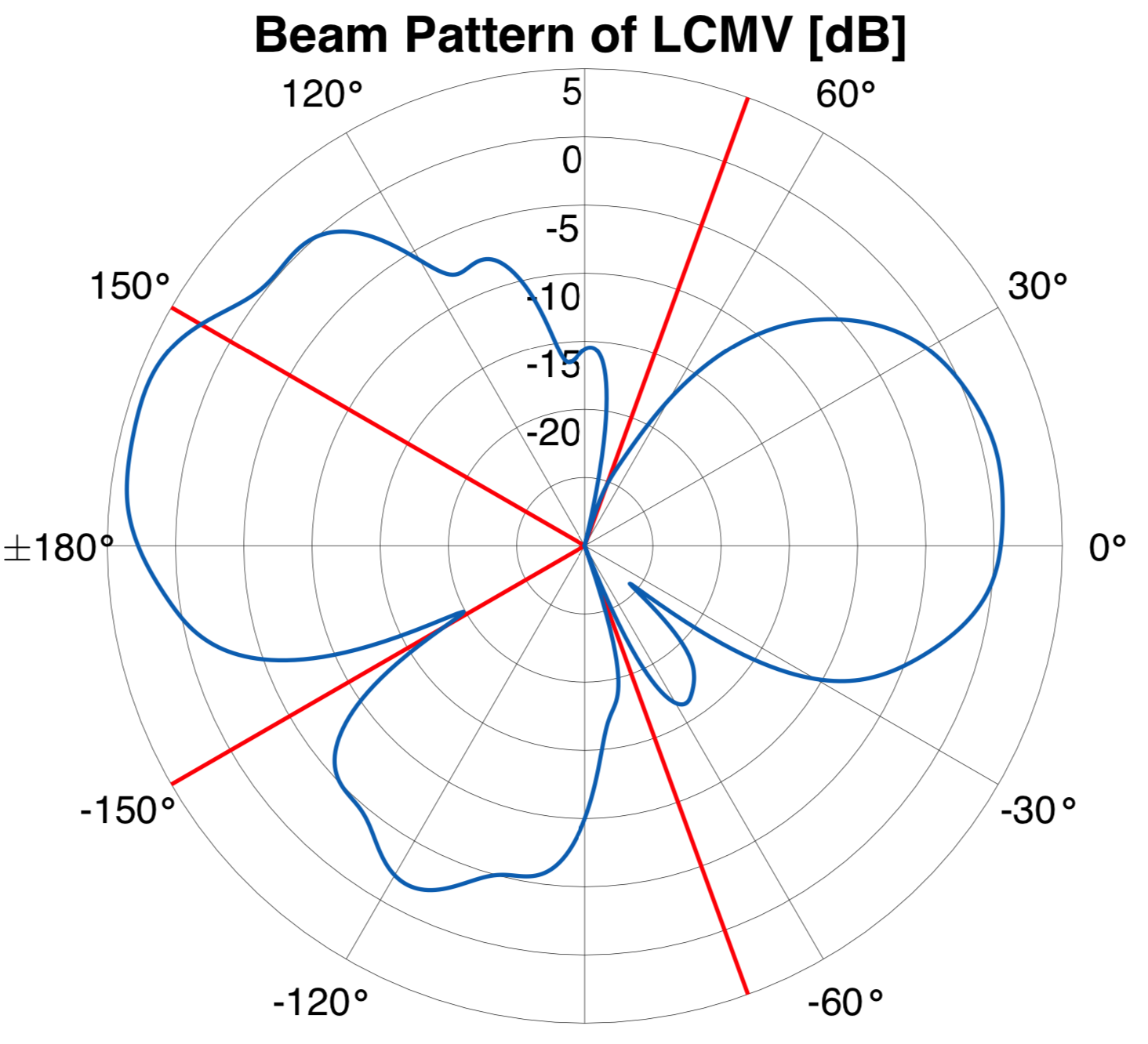}}
		\subfigure[Setup 3]{\includegraphics [width=0.23\linewidth]{ 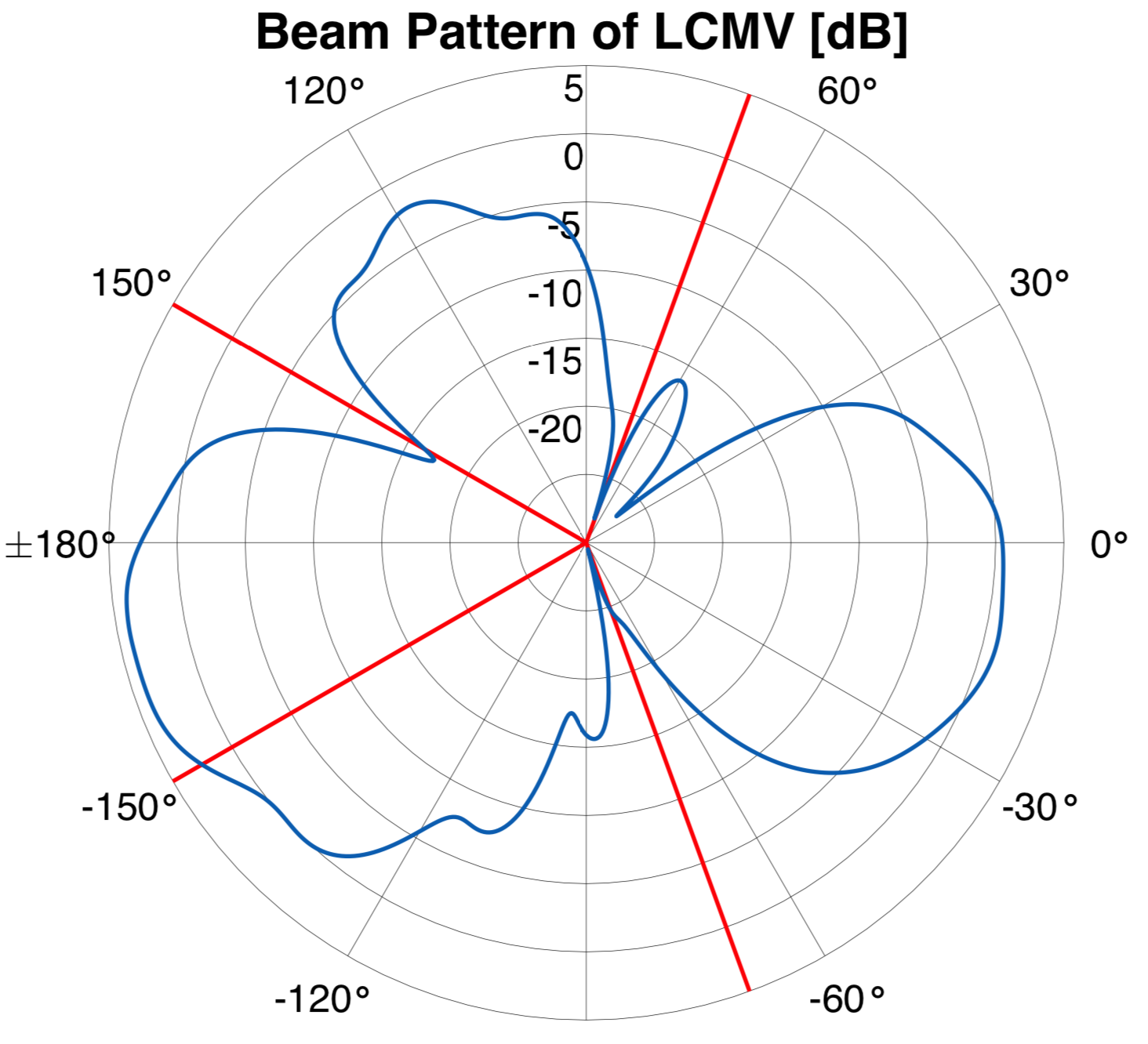}}
		\subfigure[Setup 4]{\includegraphics [width=0.23\linewidth]{ 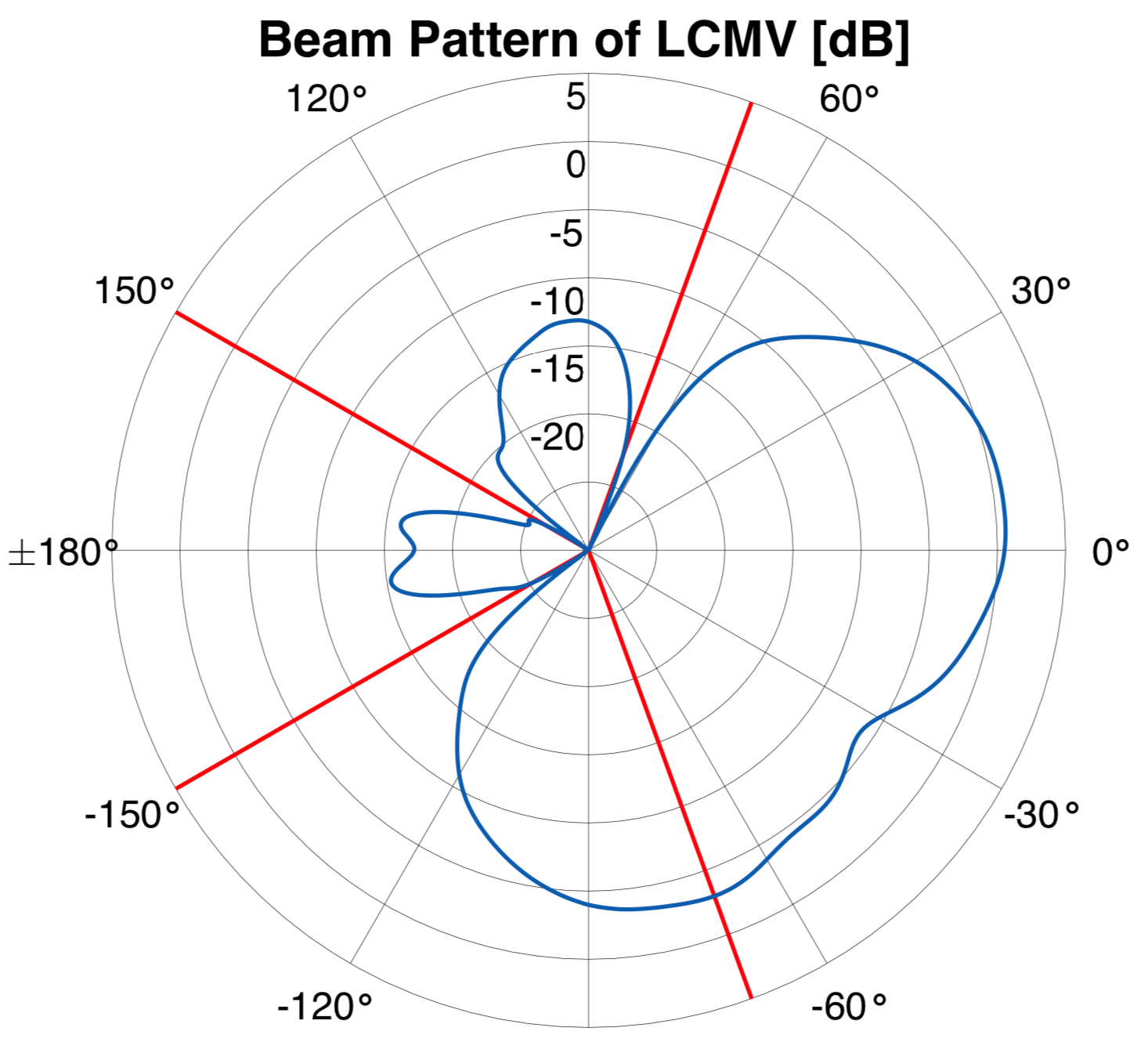}}
		\caption{Beam patterns of LCMV at 1000 Hz.}
		\label{fig:simul_lcmv}
	\end{figure}

	\begin{figure}[h]
		\subfigure[Setup 1]{\includegraphics [width=0.23\linewidth]{ 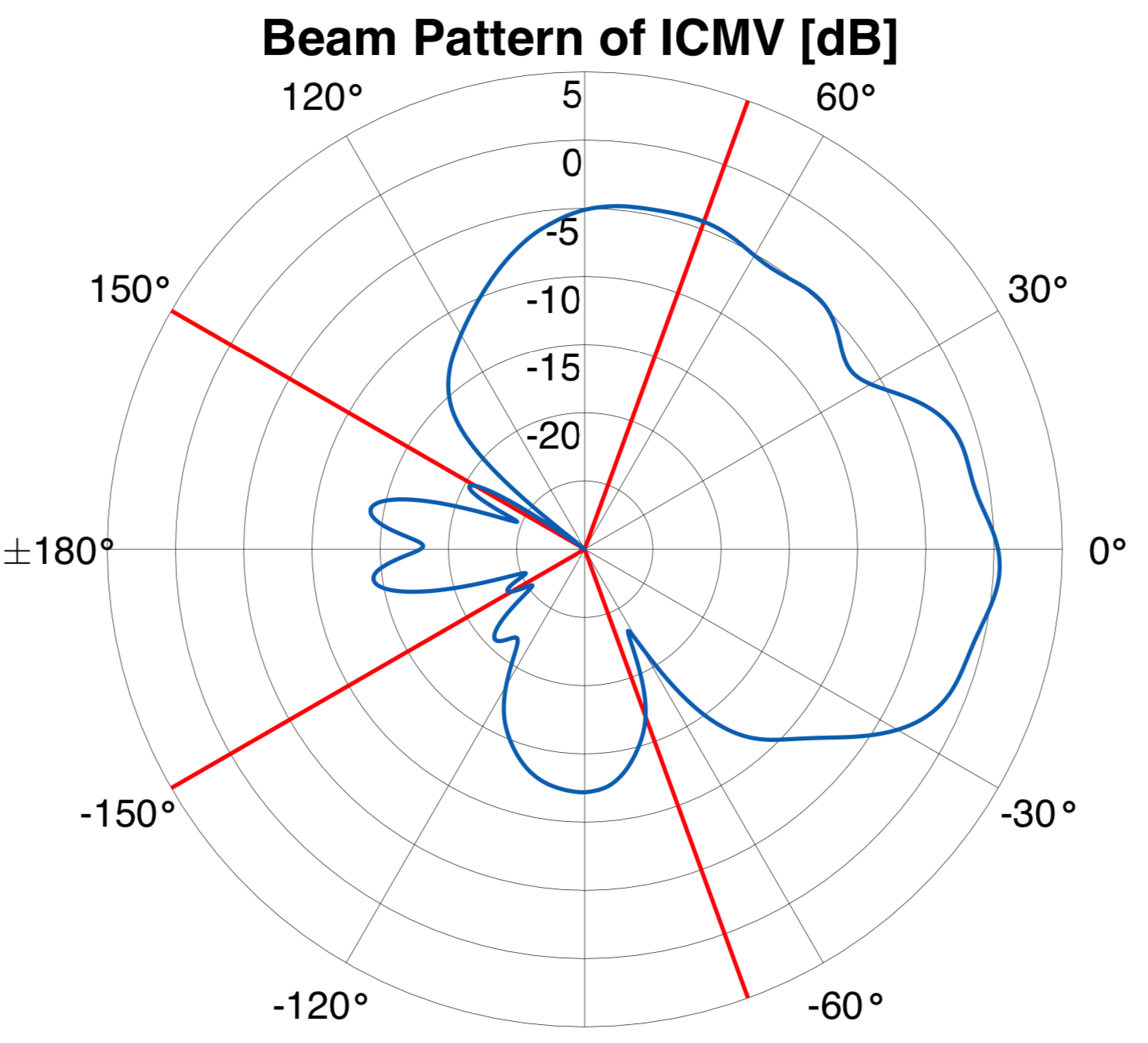}}
		\subfigure[Setup 2]{\includegraphics [width=0.23\linewidth]{ 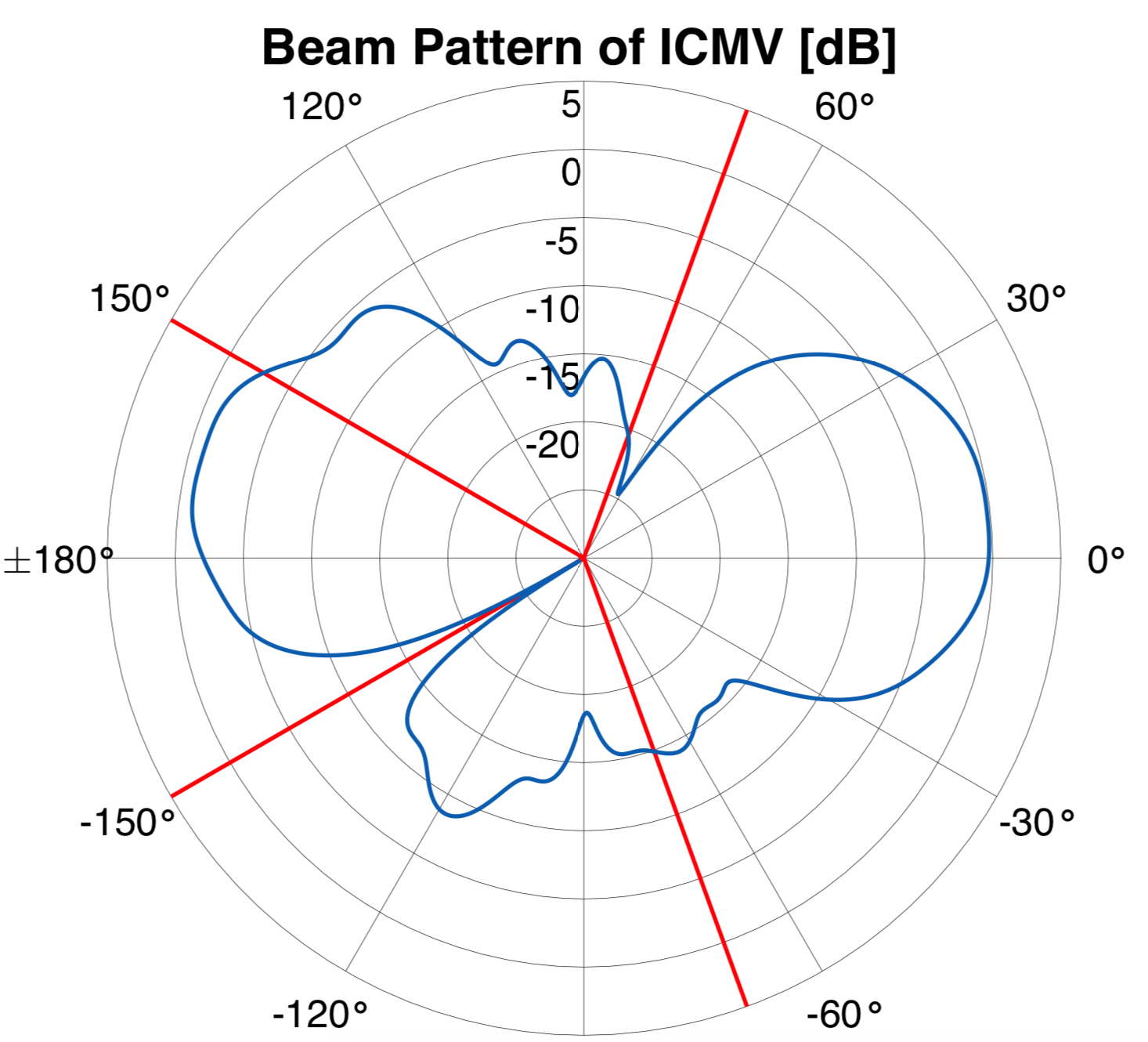}}
		\subfigure[Setup 3]{\includegraphics [width=0.23\linewidth]{ 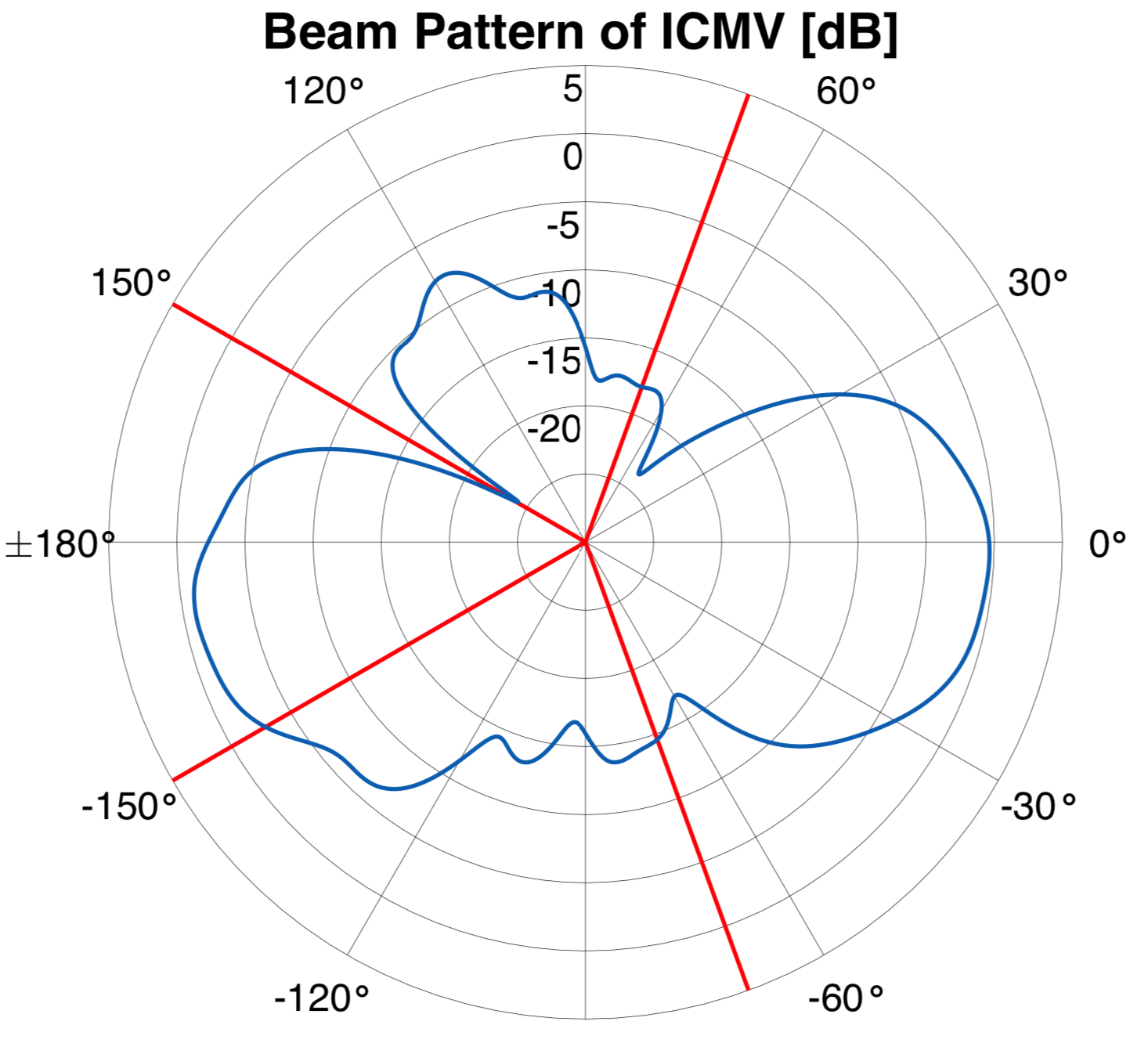}}
		\subfigure[Setup 4]{\includegraphics [width=0.23\linewidth]{ 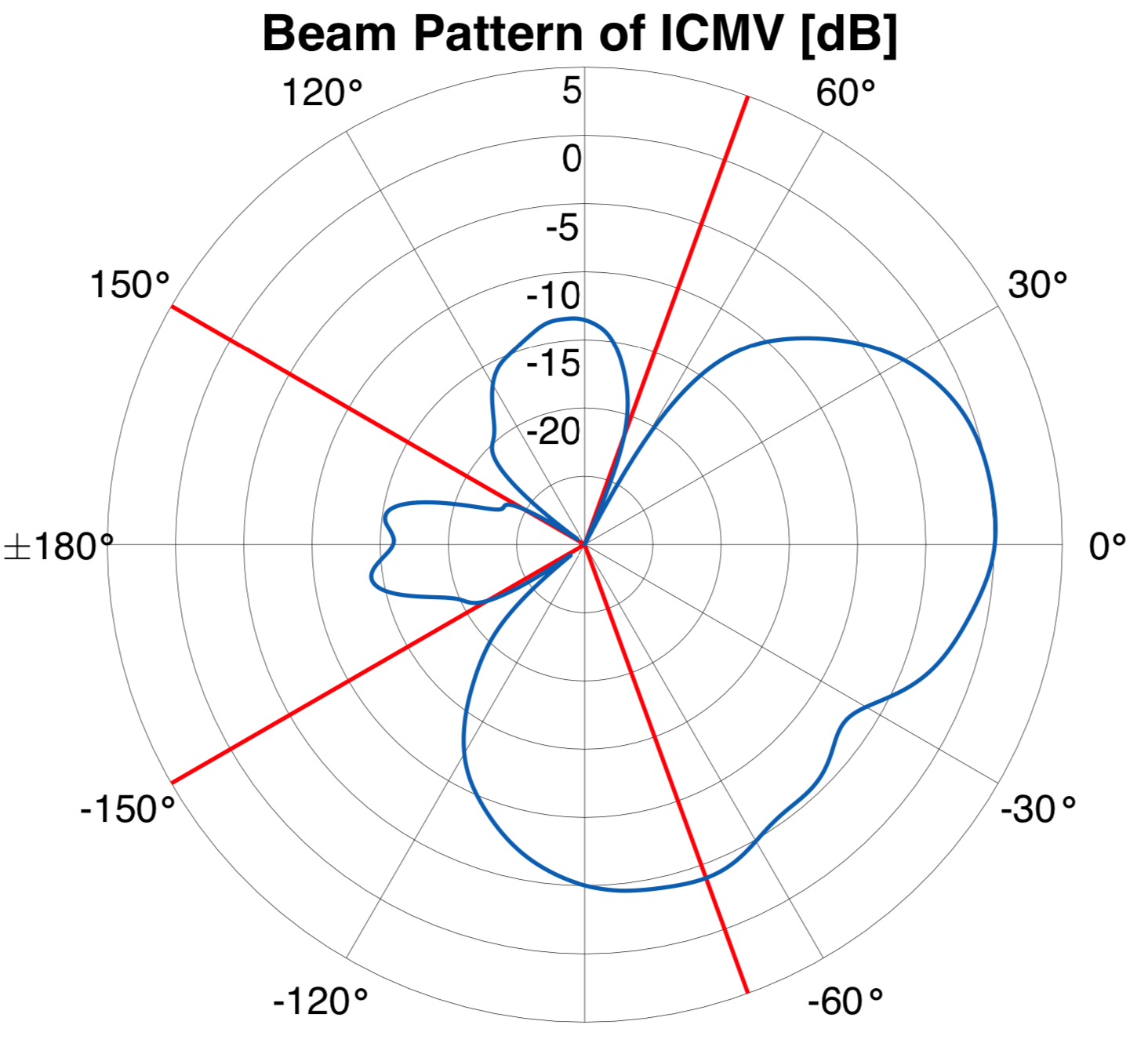}}
		\caption{Beam patterns of ICMV at 1000 Hz.}
		\label{fig:simul_icmv}
	\end{figure}

	Figs.~\ref{fig:simul_picmv},~\ref{fig:simul_lcmv}, and~\ref{fig:simul_icmv} plot the beam patterns of the three beamformers at 1000 Hz, where red lines correspond to the 4 interference directions. It can be observed that the P-ICMV beamformer has low gain in the spatial responses at all 4 interferences' direction. For LCMV and ICMV beamformers, the ignored interference direction ($\pm70^{\circ}$) has reasonable gain control due to the target constraint, but in setups 2 and 3, the ignored interference direction ($\pm150^{\circ}$) remains high spatial responses which are even larger than $0$dB. In short, when the DoF is limited, the P-ICMV beamformer can automatically handle the target source and all 4 interferences by intelligently allocating the DoF, while the LCMV and ICMV beamformers face an interference suppression selection problem. Their performance is uncontrollable and  depends on setups. {\black Furthermore, real listening evaluation on the P-ICMV beamformer is studied in \cite{Xiao2018Evaluation}, where 12 subjects' listening evaluation demonstrated that the P-ICMV beamformer can significantly improve speech intelligibility in the DoF limited situation.}
	\begin{figure}[h]
	    \centering
		\includegraphics [width=0.6\linewidth]{ 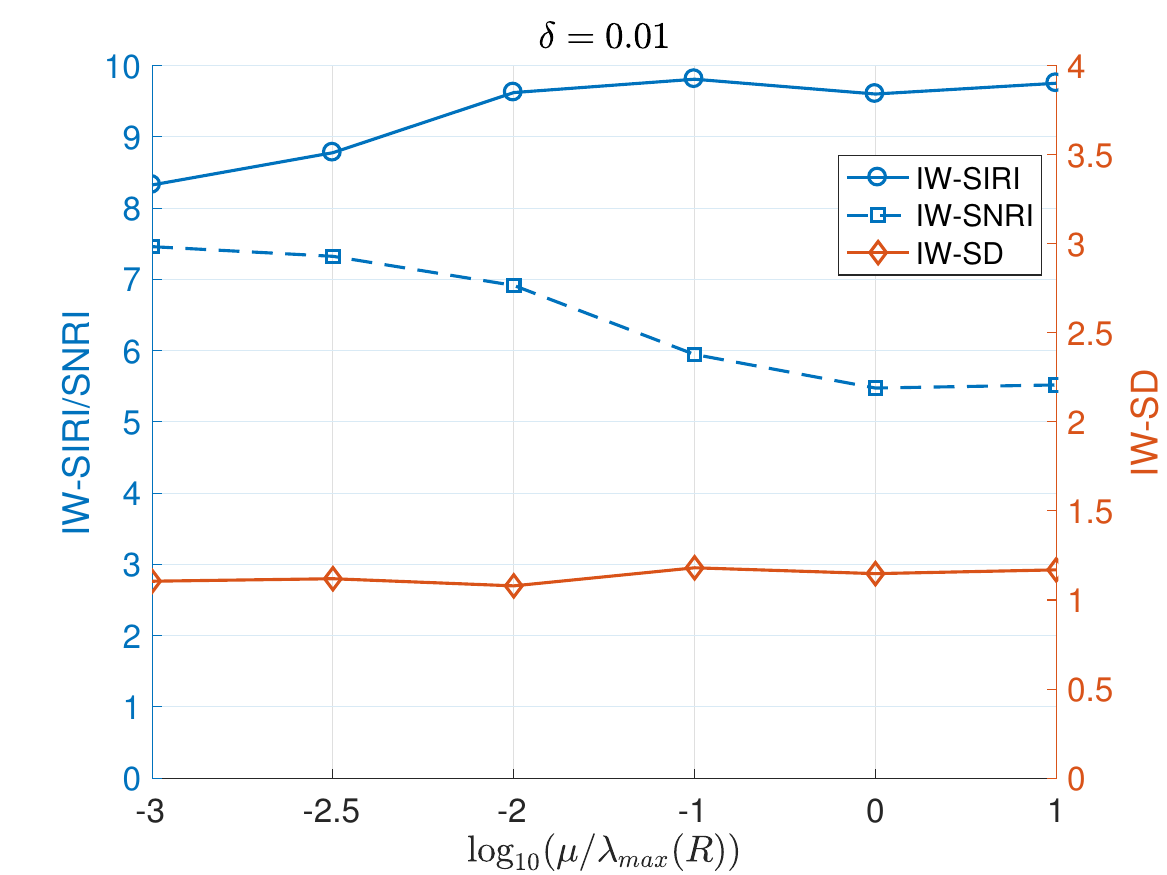}
		\caption{IW-SIRIs/SNRIs and IW-SDs with diff. values of $\mu$.}
		\label{fig:diff_mu}
	\end{figure}
	\begin{figure}[h]
	    \centering
		\includegraphics [width=0.6\linewidth]{ 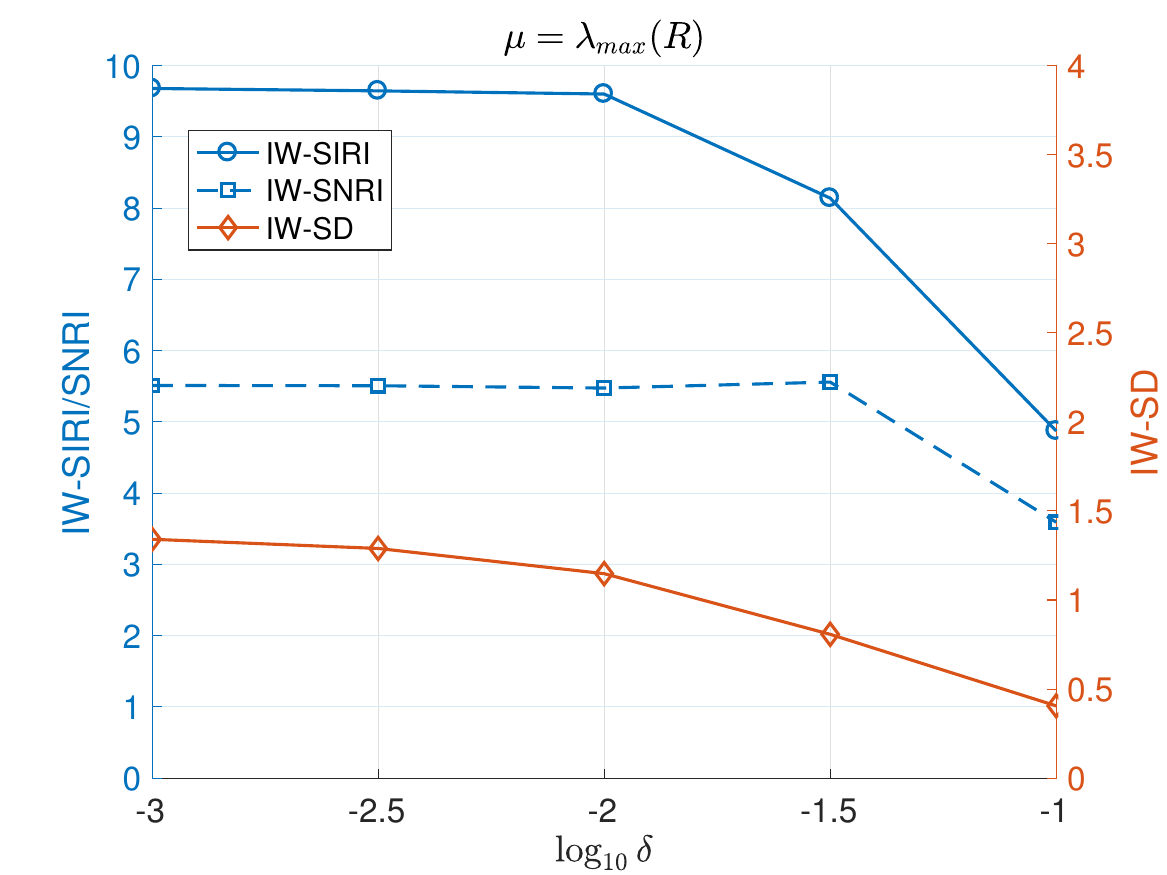}
		\caption{IW-SIRIs/SNRIs and IW-SDs with diff. values of $\delta$.}
		\label{fig:diff_delta}
	\end{figure}
	
	At last, the performance of the P-ICMV beamformer with different values of $\mu$ and $\delta$ (with the other one been fixed) are studied in Figs.~\ref{fig:diff_mu} and~\ref{fig:diff_delta}. To clearly understand the impact of these two parameters, we separately plot IW-SIRI (signal to interference ratio improvement), IW-SNRI (signal-to-noise ratio improvement), and IW-SD. In Fig.~\ref{fig:diff_mu} ($\delta=0.01$), increasing $\mu$ gives more priority to interferences suppression. For a fixed $\mu=\lambda_{\textrm{max}}(\mathbf{R})$ with different $\delta$ (Fig.~\ref{fig:diff_delta}), we can observe that increasing $\delta$ can reduce IW-SD. On the other hand, allocating more DoF to handle SV mismatch would sacrifice the interference and noise suppression performance and hence we observe IW-SIRI and IW-SNRI decreases as $\delta$ increases.

	\subsection{Beam Pattern Synthesis}\label{subsec:synthesis}
	
	We consider a $30\times30$ planar antenna array with {\black half-wave} length spacing. {\black A beam pattern with low side lobes over a large region is considered to be synthesized.} This scenario is motivated from industry applications, where non-stationary clutters/interferences comes from a low-altitude region. The beam pattern is desired to have low sidelobes within such a region. Specifically, we use $\vartheta\in[0^\circ,180^\circ]$ and $\psi\in[-90^\circ,90^\circ]$ to denote the azimuth and elevation angles. The main lobe of the specified beam pattern is pointed at $\theta_0=(\vartheta_0,\psi_0)=(90^\circ,15^\circ)$, and the side lobes (SL) within region $\Phi=\{  (\vartheta,\psi)| 0^\circ\leq\vartheta\leq180^\circ,-90^\circ\leq\psi \leq -10^\circ \}$ are expected to be as small as possible. To achieve such design requirement, we specify the P-ICMV formulation as follows, 
	\begin{subequations}\label{eq:simul_bp_formul}
		\begin{align}
		\min_{\mathbf{w},\{\epsilon_\phi\}}\ & \|  \mathbf{w}\|^2 + \mu \max_{\phi\in\Phi}\{ \epsilon_\phi \}\notag\\
		\textrm{s.t.}\ &|\mathbf{w}^H\mathbf{a}_\theta-1| + \delta \| \mathbf{w} \| \leq c_\theta,\ \forall \theta\in\Theta,\\
		& |\mathbf{w}^H\mathbf{a}_\phi| + \delta \| \mathbf{w} \| \leq \epsilon_\phi c_\phi,\ \forall \phi\in\Phi, \label{eq:simul_bp_formul2}
		\end{align}
	\end{subequations}
	where $\Theta=\{  (\vartheta,\psi)| \vartheta_0-1^\circ\leq\vartheta\leq\vartheta_0+1^\circ,\psi_0-1^\circ\leq\psi\leq\psi_0+1^\circ\}$ and angle sets $\Theta$ and $\Phi$ are sampled every $1^\circ$ for both azimuth and elevation angles. 
	
	In the simulation, {\black parameter $c_\theta$ is specified to be proportional to angle error as $c_\theta=0.3\times(\lceil\vartheta-\vartheta_0\rceil+\lceil\psi-\psi_0\rceil+1), \forall \theta\in\Theta$, where $\lceil \cdot\rceil$ denotes the operation for counting angles difference by unit degree, i.e., $\lceil \pm 1^\circ\rceil=1$.} The sufficient condition in Proposition \ref{prop:suffi_fea} gives the bound $\delta\leq 1.15$. {\black Since all sidelobes within $\Phi$ are considered to be suppressed, $c_\phi$ is fixed to be $0.1$ for all $\phi\in\Phi$.} To study the impact of parameters $\mu$ and $\delta$, we fix $\mu=10$ and evaluate the synthesized beam patterns with different choice of $\delta$ ($\delta=10^{-3},10^{-2},10^{-1},10^{-0.5}$). Notice that there are totally $14670$ constraints for the specified P-ICMV formulation~\eqref{eq:simul_bp_formul}, i.e., $3\times3=9$ constraints for $\Theta$ and $181\times81=14661$ constraints for $\Phi$. Parameter $\rho$ in ADMM algorithm is set to be $10^2$, and it only takes few tens of seconds for ADMM algorithm to numerically converge.
	\vspace{-0.5em}
	\begin{figure}[H]
		\includegraphics [width=0.23\linewidth]{ 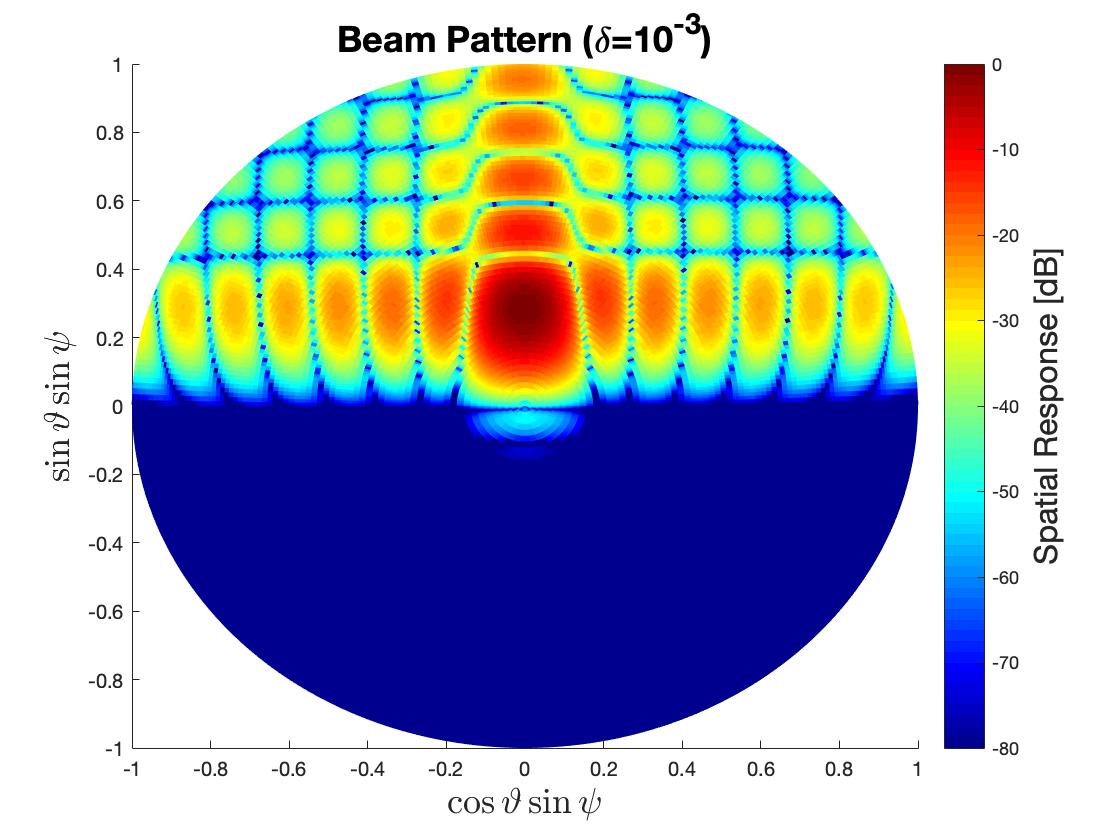}
		\includegraphics [width=0.23\linewidth]{ 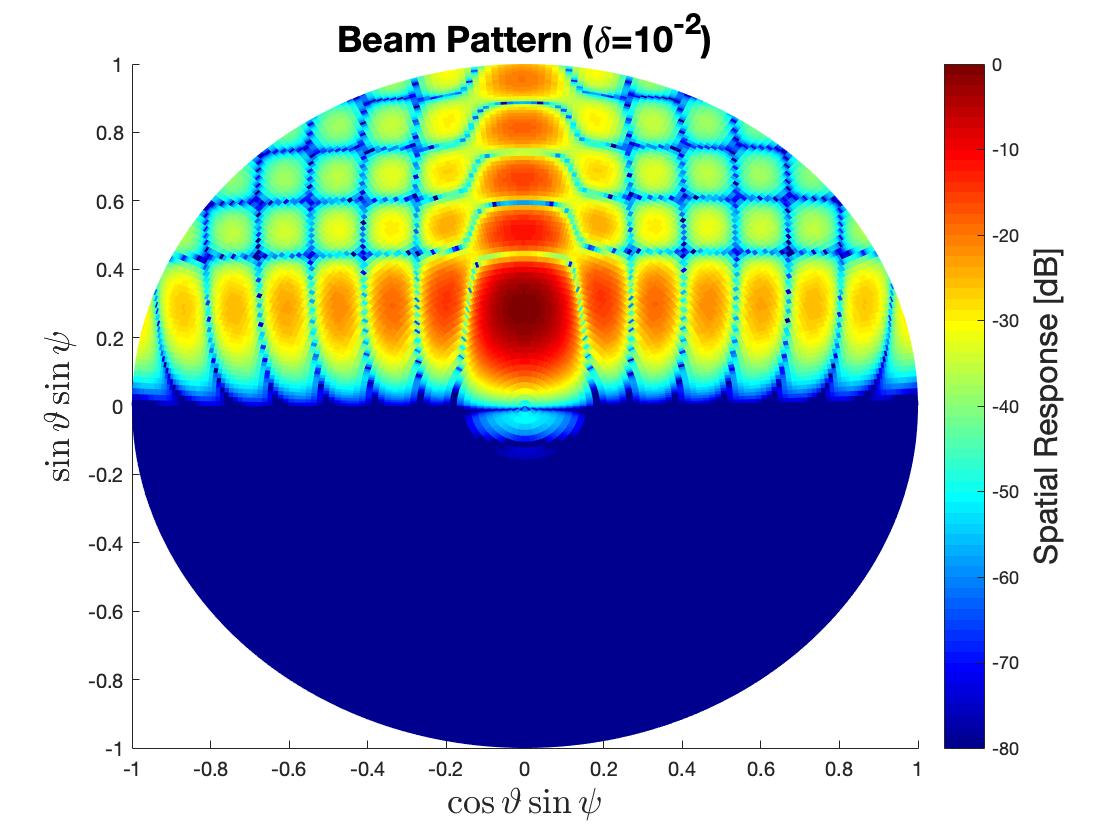}
		\includegraphics [width=0.23\linewidth]{ 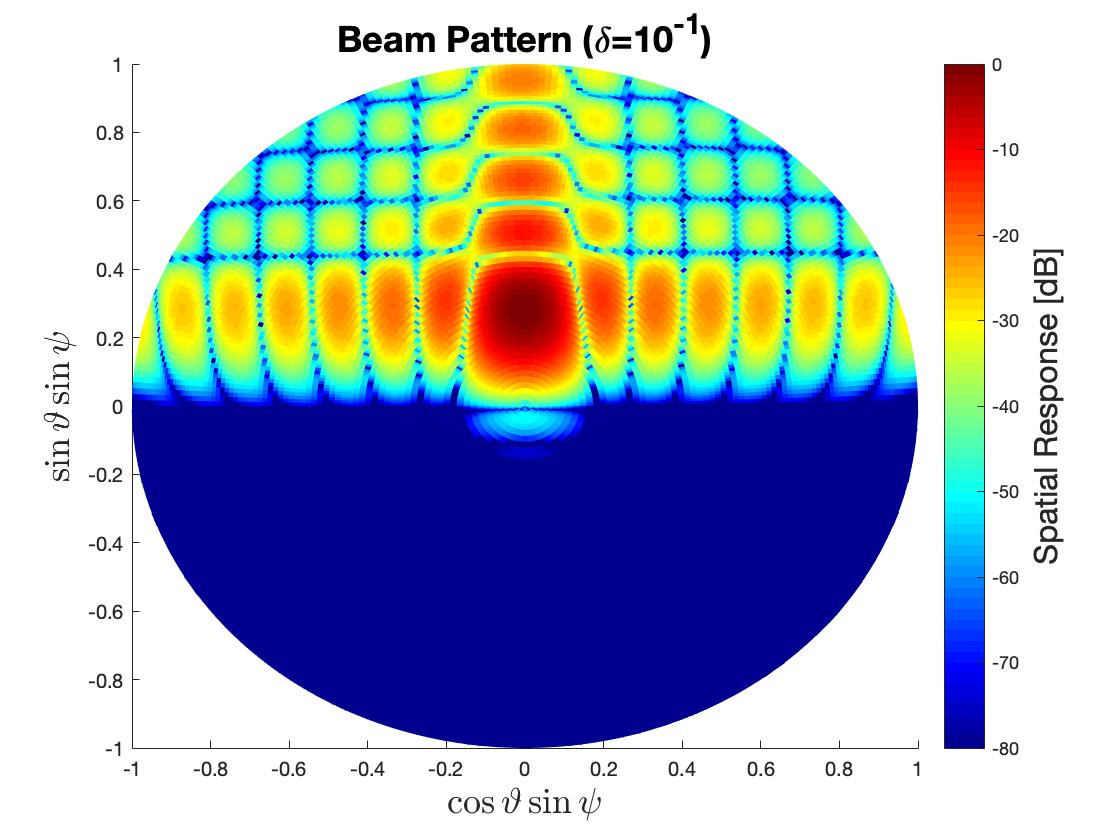}
		\includegraphics [width=0.23\linewidth]{ 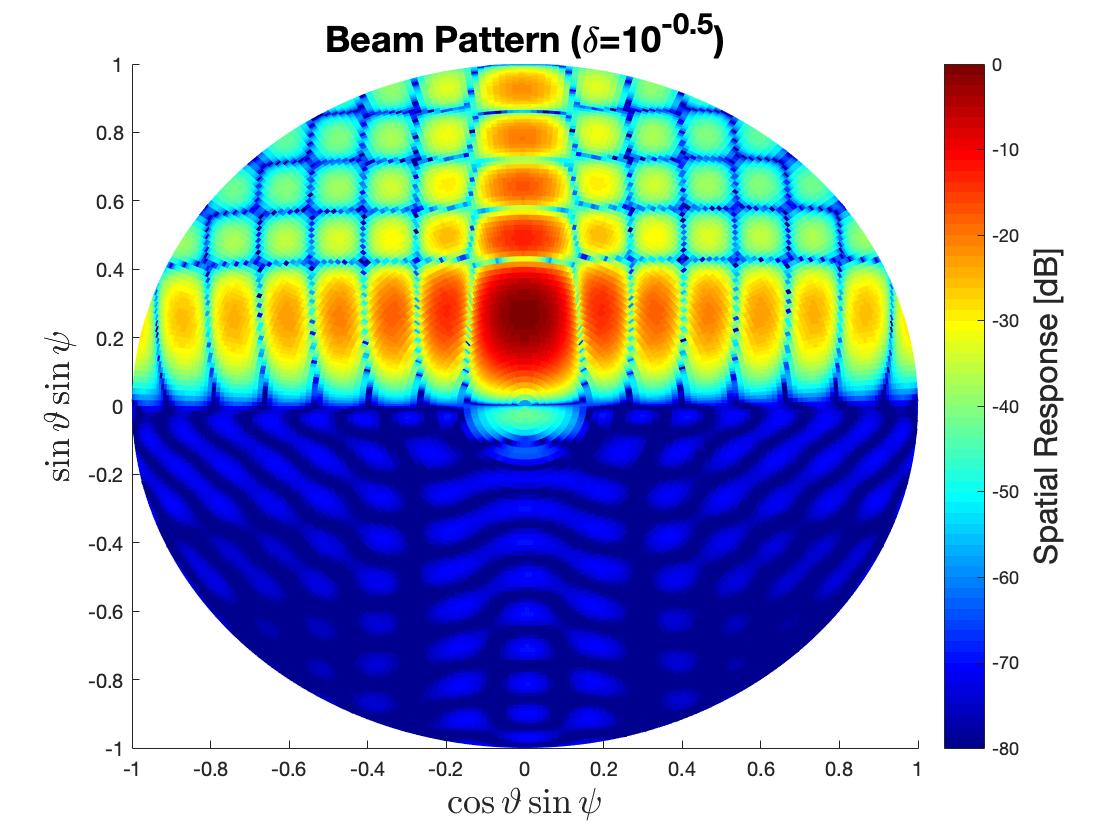}
		\caption{Synthesized beam patterns with different $\delta$.}
		\label{fig:simul_synsbp}
	\end{figure}
	
	The synthesized beam patterns with presumed SVs are plotted in Fig.~\ref{fig:simul_synsbp}. It can be observed that SL within region $\Phi$ have been well suppressed for all choice of $\delta$. However, as compared in Table~\ref{tab:ssl}, the maximum SL (MSL) and average SL (ASL) within region $\Phi$ with different $\delta$ are quite different. Small $\delta$ allows more DoF for achieving smaller side lobes but at the cost of increasing $\| \mathbf{w}\|$. For large value of $\delta$, more DoF is allocated to counter the SV mismatch,  and hence large  MSL and ASL are observed. Furthermore, in the presence of SV mismatch, the levels of true SL are affected by SV mismatch level. To evaluate the impact of $\delta$ on the levels of true SL, we assume the SV perturbation consists of {\black gain and phase errors,} which are randomly generated from $\mathcal{N}(1, \kappa^2)$ and $\mathcal{N}(0, (\kappa\pi/2)^2)$ respectively, and $\kappa>0$ controls the level of SV perturbation. The MSL and ASL with respect to $\delta=10^{-3},10^{-0.5}$ and different $\kappa$ are compared in Fig.~\ref{fig:simul_synslobe}, where the MSL and ASL are calculated by averaging $100$ independent simulation runs by randomly generating the mismatched SVs. It can be observed that small $\delta$ gives lower MSL and ASL when $\kappa$ is small ($\kappa=10^{-4}$). As $\kappa$ increases, large $\delta$ achieves lower MSL and ASL. By Proposition~\ref{prop:intf_bound}, we know constraints \eqref{eq:simul_bp_formul2} imply that the true spatial response $|\mathbf{w}^H\bar{\mathbf{a}}_\phi|\leq \epsilon_\phi c_\phi ,\forall \phi\in\Phi$ for all perturbation satisfy $\|\Delta\bm{a}_\phi\|\leq \delta$. In the case of $\| \Delta\bm{a}_\phi\|=\delta^\prime>\delta$, the true spatial response is actually bounded by $\epsilon_\phi c_\phi$ and $\| \mathbf{w}\|$, i.e., $|\mathbf{w}^H\bar{\mathbf{a}}_\phi|\leq |\mathbf{w}^H{\mathbf{a}}_\phi|+\delta^\prime \| \mathbf{w} \| \leq \epsilon_\phi c_\phi + (\delta^\prime-\delta)\| \mathbf{w} \|$.
	Hence, we observe that a large $\delta$ (which has a smaller $\|\mathbf{w}\|$) achieves a smaller MSL/ASL when $\kappa$ is large.
	\vspace{-0.5em}
	\begin{table}[H]
		\centering
		\caption{MSL and ASL [dB]}
		\begin{tabular}{c c c c c}
			\toprule
			& $\delta=10^{-3}$& $\delta=10^{-2}$& $\delta=10^{-1}$& $\delta=10^{-0.5}$\\
			\midrule
			$\|\mathbf{w}\|$ & 0.0298 & 0.0297 & 0.0290 & 0.0264\\
			\textbf{MSL}& -81.7& -81.6& -80.9&-70.1\\
			\textbf{ASL}& -89.4& -89.4& -88.4&-82.8 \\
			\bottomrule
		\end{tabular}
		\label{tab:ssl}
	\end{table}
	\vspace{-0.5em}
	\begin{figure}[h]
	    \centering
		\includegraphics [width=0.6\linewidth]{ 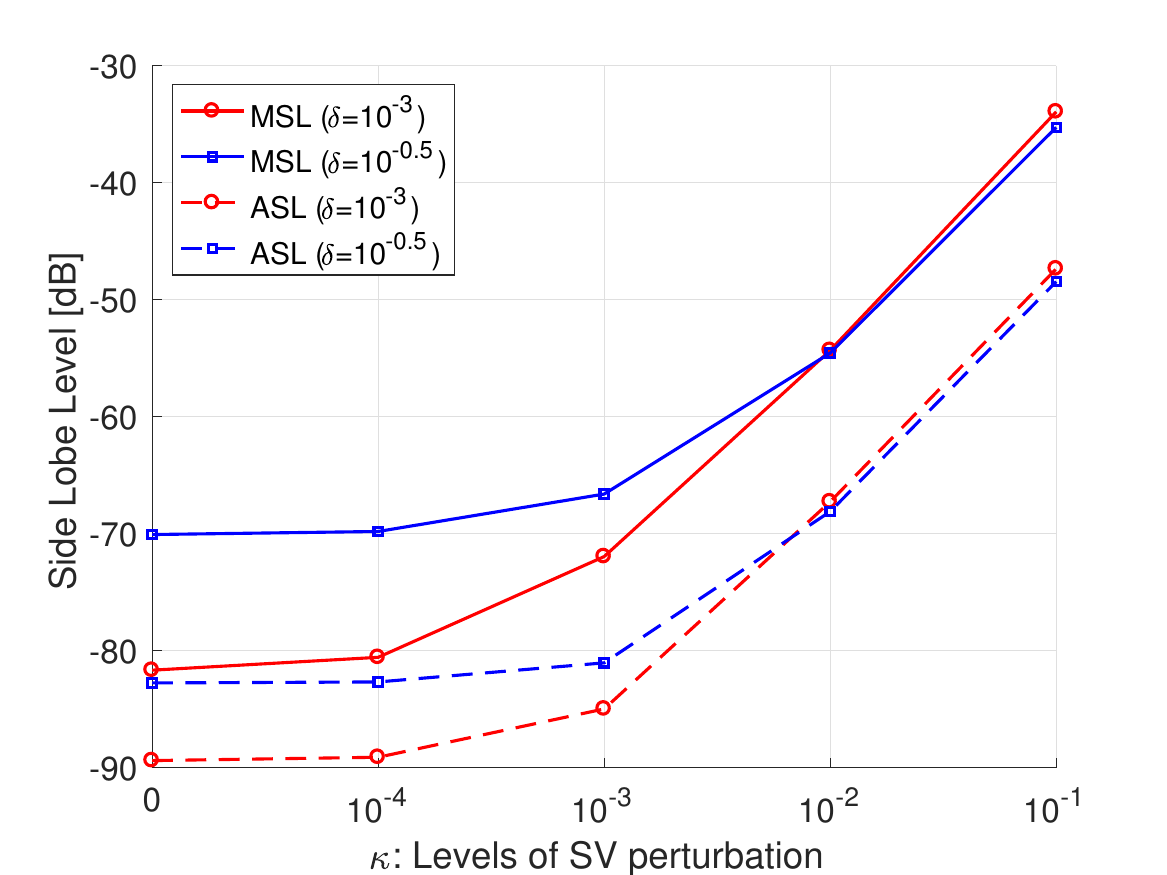}
		\caption{Side lobe levels with different $\kappa$.}
		\label{fig:simul_synslobe}
	\end{figure}
	
	\section{Conclusion}
	
	Robust beamformer design always tries to seek a balance among robustness and beamforming performance, where the central issue {\black is} how to strike an appropriate balance automatically with limited array DoF. In this paper, we propose a min-max penalization criterion for intelligently allocating the limited DoF. The proposed P-ICMV formulation makes use of two types of inequality constraints to introduce robustness against various uncertainties and a min-max penalization criterion for handling DoF limitation. Several user-specified parameters are also used in the formulation, which provide a flexible mechanism to achieve different levels of robustness. In addition, a low-complexity iterative algorithm is designed, which can compute the P-ICMV beamformer efficiently even for a large-size array. The P-ICMV beamformer can provide an effective robust solution for challenging applications where DoF is limited and model parameters are inaccurate. The ability to achieve different robustness levels is demonstrated in the simulations.

	\ifCLASSOPTIONcaptionsoff
	\newpage
	\fi

	
	\bibliographystyle{IEEEtran}
	\bibliography{refs}

\begin{thebibliography}{10}
\providecommand{\url}[1]{#1}
\csname url@samestyle\endcsname
\providecommand{\newblock}{\relax}
\providecommand{\bibinfo}[2]{#2}
\providecommand{\BIBentrySTDinterwordspacing}{\spaceskip=0pt\relax}
\providecommand{\BIBentryALTinterwordstretchfactor}{4}
\providecommand{\BIBentryALTinterwordspacing}{\spaceskip=\fontdimen2\font plus
\BIBentryALTinterwordstretchfactor\fontdimen3\font minus
  \fontdimen4\font\relax}
\providecommand{\BIBforeignlanguage}[2]{{%
\expandafter\ifx\csname l@#1\endcsname\relax
\typeout{** WARNING: IEEEtran.bst: No hyphenation pattern has been}%
\typeout{** loaded for the language `#1'. Using the pattern for}%
\typeout{** the default language instead.}%
\else
\language=\csname l@#1\endcsname
\fi
#2}}
\providecommand{\BIBdecl}{\relax}
\BIBdecl

\bibitem{Pu2017penalized}
W.~Pu, J.~Xiao, T.~Zhang, and Z.-Q. Luo, ``A penalized inequality-constrained
  minimum variance beamformer with applications in hearing aids,'' in
  \emph{2017 IEEE Workshop on Applications of Signal Processing to Audio and
  Acoustics (WASPAA)}, Oct 2017, pp. 175--179.

\bibitem{Xiao2018Evaluation}
J.~Xiao, W.~Pu, Z.-Q. Luo, and T.~Zhang, ``Evaluation of the penalized
  inequality constrained minimum variance beamformer for hearing aids,'' in
  \emph{2018 IEEE International Conference on Acoustics, Speech and Signal
  Processing (ICASSP)}, April 2018, pp. 3344--3348.

\bibitem{Conformal2020}
T.~{Cao}, W.~{Pu}, P.~{Zhang}, and Z.~{Luo}, ``Beam pattern synthesis for
  conformal array with sidelobe and polarization control: A penalized
  inequality approach,'' in \emph{2020 IEEE 11th Sensor Array and Multichannel
  Signal Processing Workshop (SAM)}, 2020, pp. 1--5.

\bibitem{tse2005fundamentals}
D.~Tse and P.~Viswanath, \emph{Fundamentals of wireless communication}.\hskip
  1em plus 0.5em minus 0.4em\relax Cambridge university press, 2005.

\bibitem{Doclo2015Magzine}
S.~Doclo, W.~Kellermann, S.~Makino, and S.~E. Nordholm, ``Multichannel signal
  enhancement algorithms for assisted listening devices: Exploiting spatial
  diversity using multiple microphones,'' \emph{IEEE Signal Processing
  Magazine}, vol.~32, no.~2, pp. 18--30, March 2015.

\bibitem{ho2013radar}
T.~Ho, J.~McWhirter, A.~Nehorai, U.~Nickel, B.~Ottersten, B.~Steinberg,
  P.~Stoica, M.~Viberg, and Z.~Zhu, \emph{Radar array processing}.\hskip 1em
  plus 0.5em minus 0.4em\relax Springer Science \& Business Media, 2013,
  vol.~25.

\bibitem{chiang2005sonar}
A.~M. Chiang and S.~R. Broadstone, ``Sonar beamforming system,'' Jan.~11 2005,
  uS Patent 6,842,401.

\bibitem{Synnevag2009Benefits}
J.~f.~Synnevag, A.~Austeng, and S.~Holm, ``Benefits of minimum-variance
  beamforming in medical ultrasound imaging,'' \emph{IEEE Transactions on
  Ultrasonics, Ferroelectrics, and Frequency Control}, vol.~56, no.~9, pp.
  1868--1879, September 2009.

\bibitem{elko1996microphone}
G.~Elko, ``Microphone array systems for hands-free telecommunication,''
  \emph{Speech communication}, vol.~20, no. 3-4, pp. 229--240, 1996.

\bibitem{mabande2009design}
E.~Mabande, A.~Schad, and W.~Kellermann, ``Design of robust superdirective
  beamformers as a convex optimization problem,'' in \emph{2009 IEEE
  International Conference on Acoustics, Speech and Signal Processing}, April
  2009, pp. 77--80.

\bibitem{Zhang2017ARC}
X.~Zhang, Z.~He, B.~Liao, X.~Zhang, Z.~Cheng, and Y.~Lu, ``$\text {A}^\text
  {2}\text {RC}$: An accurate array response control algorithm for pattern
  synthesis,'' \emph{IEEE Transactions on Signal Processing}, vol.~65, no.~7,
  pp. 1810--1824, April 2017.

\bibitem{capon1969high}
J.~Capon, ``High-resolution frequency-wavenumber spectrum analysis,''
  \emph{Proceedings of the IEEE}, vol.~57, no.~8, pp. 1408--1418, Aug 1969.

\bibitem{VOROBYOV20133264}
S.~A. Vorobyov, ``Principles of minimum variance robust adaptive beamforming
  design,'' \emph{Signal Processing}, vol.~93, no.~12, pp. 3264 -- 3277, 2013,
  {S}pecial Issue on Advances in Sensor Array Processing in Memory of Alex B.
  Gershman.

\bibitem{Buckley1987LCMV}
K.~Buckley, ``Spatial/spectral filtering with linearly constrained minimum
  variance beamformers,'' \emph{IEEE Transactions on Acoustics, Speech, and
  Signal Processing}, vol.~35, no.~3, pp. 249--266, March 1987.

\bibitem{chakrabarty2017bayesian}
S.~Chakrabarty and E.~A. Habets, ``A bayesian approach to informed spatial
  filtering with robustness against doa estimation errors,'' \emph{IEEE/ACM
  Transactions on Audio, Speech, and Language Processing}, vol.~26, no.~1, pp.
  145--160, 2017.

\bibitem{Vorobyov2003worst}
S.~A. Vorobyov, A.~B. Gershman, and Z.-Q. Luo, ``Robust adaptive beamforming
  using worst-case performance optimization: a solution to the signal mismatch
  problem,'' \emph{IEEE Transactions on Signal Processing}, vol.~51, no.~2, pp.
  313--324, Feb. 2003.

\bibitem{Li2003Capon}
J.~Li, P.~Stoica, and Z.~Wang, ``On robust capon beamforming and diagonal
  loading,'' \emph{IEEE Transactions on Signal Processing}, vol.~51, no.~7, pp.
  1702--1715, July 2003.

\bibitem{Lorenz2005Robust}
R.~G. Lorenz and S.~P. Boyd, ``Robust minimum variance beamforming,''
  \emph{IEEE Transactions on Signal Processing}, vol.~53, no.~5, pp.
  1684--1696, May 2005.

\bibitem{Nai2011Iterative}
S.~E. Nai, W.~Ser, Z.~L. Yu, and H.~Chen, ``Iterative robust minimum variance
  beamforming,'' \emph{IEEE Transactions on Signal Processing}, vol.~59, no.~4,
  pp. 1601--1611, April 2011.

\bibitem{Huang2018}
Y.~Huang, M.~Zhou, and S.~A. Vorobyov, ``{New designs on MVDR robust adaptive
  beamforming based on optimal steering vector estimation},'' \emph{arXiv},
  vol.~67, no.~14, pp. 3624--3638, 2018.

\bibitem{liao2017robust}
B.~Liao, C.~Guo, L.~Huang, Q.~Li, and H.~C. So, ``Robust adaptive beamforming
  with precise main beam control,'' \emph{IEEE Transactions on Aerospace and
  Electronic Systems}, vol.~53, no.~1, pp. 345--356, 2017.

\bibitem{Chang1992eigenspace}
L.~Chang and C.~Yeh, ``Performance of dmi and eigenspace-based beamformers,''
  \emph{IEEE Transactions on Antennas and Propagation}, vol.~40, no.~11, pp.
  1336--1347, Nov. 1992.

\bibitem{Feldman1996projection}
D.~D. Feldman, ``An analysis of the projection method for robust adaptive
  beamforming,'' \emph{IEEE Transactions on Antennas and Propagation}, vol.~44,
  no.~7, pp. 1023--1030, July 1996.

\bibitem{HUANG20121758}
F.~Huang, W.~Sheng, and X.~Ma, ``Modified projection approach for robust
  adaptive array beamforming,'' \emph{Signal Processing}, vol.~92, no.~7, pp.
  1758 -- 1763, 2012.

\bibitem{ruan2016robust}
H.~Ruan and R.~C. de~Lamare, ``Robust adaptive beamforming based on low-rank
  and cross-correlation techniques,'' \emph{IEEE Transactions on Signal
  Processing}, vol.~64, no.~15, pp. 3919--3932, 2016.

\bibitem{Cox1987Robust}
H.~Cox, R.~Zeskind, and M.~Owen, ``Robust adaptive beamforming,'' \emph{IEEE
  Transactions on Acoustics, Speech, and Signal Processing}, vol.~35, no.~10,
  pp. 1365--1376, October 1987.

\bibitem{Carlson1988Covariance}
B.~D. Carlson, ``Covariance matrix estimation errors and diagonal loading in
  adaptive arrays,'' \emph{IEEE Transactions on Aerospace and Electronic
  Systems}, vol.~24, no.~4, pp. 397--401, July 1988.

\bibitem{Elnashar2006Diagonal}
A.~Elnashar, S.~M. Elnoubi, and H.~A. El-Mikati, ``Further study on robust
  adaptive beamforming with optimum diagonal loading,'' \emph{IEEE Transactions
  on Antennas and Propagation}, vol.~54, no.~12, pp. 3647--3658, Dec. 2006.

\bibitem{Gu2012Reconstruction}
Y.~Gu and A.~Leshem, ``Robust adaptive beamforming based on interference
  covariance matrix reconstruction and steering vector estimation,'' \emph{IEEE
  Transactions on Signal Processing}, vol.~60, no.~7, pp. 3881--3885, July
  2012.

\bibitem{Ruan2014Shrinkage}
H.~Ruan and R.~C. de~Lamare, ``Robust adaptive beamforming using a
  low-complexity shrinkage-based mismatch estimation algorithm,'' \emph{IEEE
  Signal Processing Letters}, vol.~21, no.~1, pp. 60--64, Jan 2014.

\bibitem{Huang2015Reconstruction}
L.~Huang, J.~Zhang, X.~Xu, and Z.~Ye, ``Robust adaptive beamforming with a
  novel interference-plus-noise covariance matrix reconstruction method,''
  \emph{IEEE Transactions on Signal Processing}, vol.~63, no.~7, pp.
  1643--1650, April 2015.

\bibitem{zhang2015interference}
Z.~Zhang, W.~Liu, W.~Leng, A.~Wang, and H.~Shi, ``Interference-plus-noise
  covariance matrix reconstruction via spatial power spectrum sampling for
  robust adaptive beamforming,'' \emph{IEEE Signal Processing Letters},
  vol.~23, no.~1, pp. 121--125, 2015.

\bibitem{yang2018high}
L.~Yang, M.~R. McKay, and R.~Couillet, ``High-dimensional mvdr beamforming:
  Optimized solutions based on spiked random matrix models,'' \emph{IEEE
  Transactions on Signal Processing}, vol.~66, no.~7, pp. 1933--1947, 2018.

\bibitem{bertsekas1999nonlinear}
D.~P. Bertsekas, \emph{Nonlinear programming}.\hskip 1em plus 0.5em minus
  0.4em\relax Athena scientific Belmont, 1999.

\bibitem{van2004optimum}
H.~L. Van~Trees, \emph{Optimum array processing: Part IV of detection,
  estimation, and modulation theory}.\hskip 1em plus 0.5em minus 0.4em\relax
  John Wiley \& Sons, 2004.

\bibitem{Widrow1982cancellation}
B.~Widrow, K.~Duvall, R.~Gooch, and W.~Newman, ``Signal cancellation phenomena
  in adaptive antennas: Causes and cures,'' \emph{IEEE Transactions on Antennas
  and Propagation}, vol.~30, no.~3, pp. 469--478, May 1982.

\bibitem{marquardt2015interaural}
D.~Marquardt, V.~Hohmann, and S.~Doclo, ``Interaural coherence preservation in
  multi-channel wiener filtering-based noise reduction for binaural hearing
  aids,'' \emph{IEEE/ACM Transactions on Audio, Speech and Language Processing
  (TASLP)}, vol.~23, no.~12, pp. 2162--2176, 2015.

\bibitem{Hadad2016Theoretical}
E.~Hadad, D.~Marquardt, S.~Doclo, and S.~Gannot, ``Theoretical analysis of
  binaural transfer function mvdr beamformers with interference cue
  preservation constraints,'' \emph{IEEE/ACM Transactions on Audio, Speech, and
  Language Processing}, vol.~23, no.~12, pp. 2449--2464, Dec 2015.

\bibitem{Liao2015bcd}
W.~C. Liao, M.~Hong, I.~Merks, T.~Zhang, and Z.-Q. Luo, ``Incorporating spatial
  information in binaural beamforming for noise suppression in hearing aids,''
  in \emph{2015 IEEE International Conference on Acoustics, Speech and Signal
  Processing (ICASSP)}, April 2015, pp. 5733--5737.

\bibitem{Liao2016admm}
W.~C. Liao, Z.-Q. Luo, I.~Merks, and T.~Zhang, ``An effective low complexity
  binaural beamforming algorithm for hearing aids,'' in \emph{2015 IEEE
  Workshop on Applications of Signal Processing to Audio and Acoustics
  (WASPAA)}, Oct 2015, pp. 1--5.

\bibitem{ye2011interior}
Y.~Ye, \emph{Interior point algorithms: theory and analysis}.\hskip 1em plus
  0.5em minus 0.4em\relax John Wiley \& Sons, 2011, vol.~44.

\bibitem{Hong2017}
M.~Hong and Z.-Q. Luo, ``On the linear convergence of the alternating direction
  method of multipliers,'' \emph{Mathematical Programming}, vol. 162, no.~1,
  pp. 165--199, 2017.

\bibitem{boyd2011distributed}
S.~Boyd, N.~Parikh, E.~Chu, B.~Peleato, and J.~Eckstein, ``Distributed
  optimization and statistical learning via the alternating direction method of
  multipliers,'' \emph{Foundations and Trends{\textregistered} in Machine
  Learning}, vol.~3, no.~1, pp. 1--122, 2011.

\bibitem{fan2019robust}
W.~Fan, J.~Liang, G.~Yu, H.-C. So, and J.~Li, ``Robust capon beamforming via
  admm,'' in \emph{ICASSP 2019-2019 IEEE International Conference on Acoustics,
  Speech and Signal Processing (ICASSP)}.\hskip 1em plus 0.5em minus
  0.4em\relax IEEE, 2019, pp. 4345--4349.

\bibitem{liang2018sparse}
J.~Liang, X.~Zhang, H.~C. So, and D.~Zhou, ``Sparse array beampattern synthesis
  via alternating direction method of multipliers,'' \emph{IEEE Transactions on
  Antennas and Propagation}, vol.~66, no.~5, pp. 2333--2345, 2018.

\bibitem{cheng2017constant}
Z.~Cheng, Z.~He, S.~Zhang, and J.~Li, ``Constant modulus waveform design for
  mimo radar transmit beampattern,'' \emph{IEEE Transactions on Signal
  Processing}, vol.~65, no.~18, pp. 4912--4923, 2017.

\bibitem{yu2020quadratic}
X.~Yu, G.~Cui, J.~Yang, J.~Li, and L.~Kong, ``Quadratic optimization for
  unimodular sequence design via an adpm framework,'' \emph{IEEE Transactions
  on Signal Processing}, vol.~68, pp. 3619--3634, 2020.

\bibitem{gemechu2019beampattern}
A.~Y. Gemechu, G.~Cui, X.~Yu, and L.~Kong, ``Beampattern synthesis with
  sidelobe control and applications,'' \emph{IEEE Transactions on Antennas and
  Propagation}, vol.~68, no.~1, pp. 297--310, 2019.

\bibitem{feng2020phased}
L.~Feng, G.~Cui, X.~Yu, Z.~Zhang, and L.~Kong, ``Phased array beamforming with
  practical constraints,'' \emph{Signal Processing}, vol. 176, p. 107698, 2020.

\bibitem{Khabbazibasmenj2012Estimation}
A.~Khabbazibasmenj, S.~A. Vorobyov, and A.~Hassanien, ``Robust adaptive
  beamforming based on steering vector estimation with as little as possible
  prior information,'' \emph{IEEE Transactions on Signal Processing}, vol.~60,
  no.~6, pp. 2974--2987, June 2012.

\bibitem{allen1979image}
J.~B. Allen and D.~A. Berkley, ``Image method for efficiently simulating
  small-room acoustics,'' \emph{The Journal of the Acoustical Society of
  America}, vol.~65, no.~4, pp. 943--950, 1979.

\bibitem{garofolo1993darpa}
J.~S. Garofolo, L.~F. Lamel, W.~M. Fisher, J.~G. Fiscus, and D.~S. Pallett,
  ``{DARPA TIMIT} acoustic-phonetic continous speech corpus,'' 1993.

\bibitem{Spriet2005Robustness}
A.~Spriet, M.~Moonen, and J.~Wouters, ``Robustness analysis of multichannel
  wiener filtering and generalized sidelobe cancellation for multimicrophone
  noise reduction in hearing aid applications,'' \emph{IEEE Transactions on
  Speech and Audio Processing}, vol.~13, no.~4, pp. 487--503, July 2005.

\end{thebibliography}

\vfill\pagebreak
\newpage
\clearpage
\begin{center} 
{\large\textbf{Supplementary Material}}
\end{center}
\normalsize

\appendices
\section{Proof of Lemma 1}\label{app:QCQP_y}
The proof of Lemma 1 contains two parts. We first prove that, for fixed $y$ with $c-\delta y\geq 0$, problem (17) has closed-form solution given as $x^*=d-ce^{j\psi}+e^{j\psi}\max\{  c-r, \delta y \}$, where $\psi=\angle (2ad+b)$ and $r=|\frac{2ad+b}{2a}|$. Then we give the closed-form solution for the optimal $y^*$. 

For any fixed $y\leq c/\delta$, problem (17) becomes 
\begin{equation}\label{eq:simp_QCQP}
\min _{x\in\mathbb{C}} \ a|x|^2+{\rm{Re}}\{ b^Hx  \}\quad {\rm{s.t.}}\ |x-d|^2\leq \bar{c}^2,
\end{equation}
where $\bar{c}=c-\delta y\geq 0.$ Since problem \eqref{eq:simp_QCQP} is a strongly convex problem ($a>0$), its unique optimal solution is its KKT point. The KKT conditions of problem \eqref{eq:simp_QCQP} are 
	\begin{subequations}\label{eq:simp_QCQP_kkt}
		\begin{align}
		2ax+b+2\lambda(x-d)&=0,\label{kktQCQP:first}\\
		\lambda(|x-d|^2-\bar{c}^2)&=0,\label{kktQCQP:comp}\\
		|x-d|^2-\bar{c}^2&\leq 0,\label{kktQCQP:primal}\\
		\lambda&\geq 0\label{kktQCQP:dual},
		\end{align}
	\end{subequations}
where $\lambda$ is the Lagrangian multiplier associating with constraint $|x-d|^2\leq \bar{c}^2$. If $2ad=-b$, then $x=d$ and $\lambda=0$ satisfy all conditions in \eqref{eq:simp_QCQP_kkt} and hence is the optimal solution. Otherwise, condition \eqref{kktQCQP:first} (with $x\neq d$) is equivalent to
\begin{equation}\label{eq:lem1_kkt1}
\lambda=-\frac{2ax+b}{2(x-d)}
\end{equation}
In the next, we replace conditions \eqref{kktQCQP:first} by \eqref{eq:lem1_kkt1} and study conditions \eqref{kktQCQP:comp} and \eqref{kktQCQP:primal} with  $\lambda>0$ and $\lambda=0$ (condition \eqref{kktQCQP:dual}). If $\lambda>0$, then conditions \eqref{kktQCQP:comp} and \eqref{kktQCQP:primal} hold only when $|x-d|^2=\bar{c}^2$. This implies $x=d+\bar{c}e^{j\phi}$, where $\phi\in[-\pi,\pi]$ is a rotation angle. Substitute $x$ into \eqref{eq:lem1_kkt1}, we have 
$$
\lambda=-\frac{2a(d+\bar{c}e^{j\phi})+b}{2\bar{c}e^{j\phi}}=-\frac{2ade^{-j\phi}+2a\bar{c}+be^{-j\phi}}{2\bar{c}}.
$$
Notice $\lambda$ is a positive real number, which implies $-(2ad+b)e^{-j\phi}$ is also a positive real number. Based on the fact that rotating any complex number to be a positive real number only holds when the rotation angle takes its negative phase. Hence we have the only choice for $e^{j\phi}=-\frac{2ad+b}{|2ad+b|}$. If $\lambda=0$, then $x=-\frac{b}{2a}$ and condition \eqref{kktQCQP:primal} implies $|2ad+b|\leq 2a\bar{c}$ must hold. Combine the cases for $\lambda>0$ and $\lambda=0$, conditions \eqref{kktQCQP:first}-\eqref{kktQCQP:dual} are simplified as 
\begin{equation}\label{eq:lem1_x}
	x^*=\left\lbrace
	\begin{aligned}
	&-\frac{b}{2a},\quad \textrm{if }|2ad+b|\leq 2a\bar{c},&\\
	&d-\frac{2ad+b}{|2ad+b|}\bar{c},\quad \textrm{otherwise}.&
	\end{aligned}
	\right.
\end{equation}
Notice the special case that $2ad=-b,x=d$ is also included in \eqref{eq:lem1_x}. Compactly, define $\psi=\angle (2ad+b)$ and $r=|\frac{2ad+b}{2a}|$, we have $-\frac{b}{2a}=d-re^{j\psi}$, and hence \eqref{eq:lem1_x} can be expressed as 
\begin{equation}\label{eq:lem1_x_comp}
\begin{aligned}
x^*&=\left\lbrace
\begin{aligned}
&d-re^{j\psi},\quad \textrm{if }r\leq \bar{c},&\\
&d-\bar{c}e^{j\psi},\quad \textrm{otherwise}.&
\end{aligned}
\right.\\
&=d+e^{j\psi}\max\{  -r, -\bar{c} \}\ (\bar{c}=c-\delta y)\\
&=d-ce^{j\psi}+e^{j\psi}\max\{  c-r, \delta y \}.
\end{aligned}
\end{equation}
This completes the first part of the proof.

Now we prove the optimal $y^*$ for problem (17) also has a closed-form based on \eqref{eq:lem1_x_comp}. Substitute  \eqref{eq:lem1_x_comp} into problem (17) and ignore some constant terms, problem (17) becomes 
\begin{equation}\label{eq:lem1_single_y}
\begin{aligned}
\min_y \  f(y)\quad
\textrm{s.t.}\ y\leq c/\delta,
\end{aligned}
\end{equation}
where 
$$f(y)=a\left(  \max\{  c-r, \delta y  \} \right)^2 + 2a(r-c)\max\{ c-r, \delta y  \} + \alpha y^2 + \beta y.$$
Notice $\max\{  c-r, \delta y  \} $ is a piecewise linear convex function with respect to $y$ and non-decreasing. This implies $f(y)$ is continuous and strongly convex. Then the optimal $y^*$ for \eqref{eq:lem1_single_y} must satisfy $\partial f(y^*)=0$ or $y^*=c/\delta$, where $\partial f(y)$ is the gradient of $f(y)$ given as 
\begin{equation}
\partial f(y)=2a\delta^2\max\{  0, y-\frac{c-r}{\delta}  \}+2\alpha y + \beta .
\end{equation}
If $\partial f(\frac{c}{\delta})=\frac{2ar+2\alpha c+\delta \beta}{\delta}\leq 0$, $f(y)$ is strictly decreasing for $y\leq c/\delta$, hence $y^*=c/\delta$ lies on the boundary. Otherwise, we must have $y^*\leq\delta/c$ such that $\partial f(y^*)=0$. If $2\alpha(c-r)+\beta \delta\geq 0$, we have $\partial f(\frac{c-r}{\delta})\geq 0$, which implies  $f(y^*)=2\alpha y^*+\beta=0$ and hence $y^*=-\frac{\beta}{2\alpha}$; otherwise we have $f(y^*)=2a\delta^2y^*-2a\delta(c-r)+2\alpha y^*+\beta=0$ and $y^*=\frac{2a\delta(c-r)-\beta}{2a\delta^2+2\alpha}$. In short, define $\lambda_1=2\alpha c - 2 \alpha r+\delta\beta<\lambda_2=2\alpha c+2ar+\delta \beta$, we have 
\begin{equation}
y^*=\left\lbrace
\begin{aligned}
	&y_1=-\frac{\beta}{2\alpha},\quad \textrm{if }\lambda_1 \geq 0,&\\
	&y_2=\frac{c}{\delta},\quad \textrm{if } \lambda_2\leq 0,&\\
	&y_3=\frac{2a\delta(c-r)-\beta}{2a\delta^2+2\alpha},\quad \textrm{if }\lambda_1 < 0, \lambda_2>0.&
\end{aligned}
\right.
\end{equation}
Notice that $\lambda_1 \geq 0$ implies $y_1=\min_{i=1,2,3}\{ y_i\}$, $\lambda_2 \leq 0$ implies $y_2=\min_{i=1,2,3}\{ y_i\}$, and $\lambda_1 < 0,\lambda_2>0$ imply $y_3=\min_{i=1,2,3}\{ y_i\}$. Compactly, we have
\begin{equation}
\begin{aligned}
y^*&=\min\{ -\frac{\beta}{2\alpha},\frac{2a\delta(c-r)-\beta}{2a\delta^2+2\alpha},c/\delta  \},\\
\end{aligned}
\end{equation}
This completes the proof.
\section{Proof of Lemma 2}\label{app:2}
We first transform problem (19) into an equivalent form, and then prove Lemma 2 by studying KKT conditions of the transformed problem. Let the eigenvalue decomposition with respect to $\mathbf{A}\succ 0$ as $\mathbf{A}=\mathbf{U}\bm{\Lambda}\mathbf{U}^H$, where $\mathbf{U}$ is a unitary matrix stacked by eigenvectors and $\bm{\Lambda}$ is a diagonal matrix with diagonal elements corresponding to eigenvalues $\lambda_i>0,i=1,2,\ldots,M$. Since for any $\mathbf{w}\in\mathbb{C}^M$, $\| \mathbf{U}^H\mathbf{w}\|=\|\mathbf{w}\|$ always holds, problem (19) can be equivalently transformed as 
\begin{equation}\label{eq:lem_sub_wy}
\begin{aligned}
\min_{\mathbf{\bar{w}}}\quad&\mathbf{\bar{w}}^H\bm{\Lambda}\mathbf{\bar{w}}+\textrm{Re}\{ \mathbf{\bar{b}}^H\mathbf{\bar{w}}\} + \alpha y^2 + \beta y\\
\textrm{s.t.} \quad& \| \mathbf{\bar{w}} \| \leq y,
\end{aligned}
\end{equation}
where $\mathbf{\bar{b}}=\mathbf{U}\mathbf{b}$. If $(\mathbf{\bar{w}}^*,y^*)$ is optimal for problem \eqref{eq:lem_sub_wy}, then $(\mathbf{U}^H\mathbf{\bar{w}}^*,y^*)$ is optimal for problem (19). The Lagrangian function for problem \eqref{eq:lem_sub_wy} is 
\begin{equation*}
L(\mathbf{\bar{w}},y,\lambda)=\mathbf{\bar{w}}^H\bm{\Lambda}\mathbf{\bar{w}}+\textrm{Re}\{ \mathbf{\bar{b}}^H\mathbf{\bar{w}}\} + \alpha y^2 + \beta y+\lambda(\| \mathbf{\bar{w}} \| - y),
\end{equation*}
where $\lambda\geq 0$ is the Lagrangian multiplier associated with the constraint $\| \mathbf{\bar{w}}\| \leq y$. Notice $L(\mathbf{\bar{w}},y,\lambda)$ is not differentiable at $\mathbf{\bar{w}}=\mathbf{0}$. To study the case for $\mathbf{\bar{w}}^*\neq \mathbf{0}$, we consider $\beta<\| \mathbf{\bar{b}}\|$, which is a sufficient and necessary condition for $\mathbf{\bar{w}}^*\neq \mathbf{0}$ (see Proposition \ref{prop:suff_lem}). Then, the KKT conditions of problem \eqref{eq:lem_sub_wy} for $\beta<\| \mathbf{\bar{b}}\|$ are
\begin{subequations}
	\begin{align}
		2\bm{\Lambda}\mathbf{\bar{w}}+\mathbf{\bar{b}}+\lambda\frac{\mathbf{\bar{w}}}{\| \mathbf{\bar{w}}\|}&=\mathbf{0},\label{eq:lem_bisec_w}\\
		2\alpha y + \beta - \lambda &= 0,\label{eq:lem_bisec_y}\\
		\| \mathbf{\bar{w}} \|&\leq y,\label{eq:lem_bisec_primal} \\
		\lambda &\geq 0,\label{eq:lem_bisec_dual}\\
		\lambda(\| \mathbf{\bar{w}} \|- y )&=0,\label{eq:lem_bisec_comp}
	\end{align}
\end{subequations}
If $\lambda=0$, we have $y=-\frac{\beta}{2\alpha}$ (by \eqref{eq:lem_bisec_y}) and $\mathbf{\bar{w}}=-\frac{1}{2}\bm{\Lambda}^{-1}\mathbf{\bar{b}}$ (by \eqref{eq:lem_bisec_w}). In this case, \eqref{eq:lem_bisec_primal} holds if $\|\bm{\Lambda}^{-1}\mathbf{\bar{b}}\| \leq-\frac{\beta}{\alpha}$, otherwise we must  have $\lambda>0$. For $\lambda>0$, we have $\| \mathbf{\bar{w}} \| = y$ (by \eqref{eq:lem_bisec_comp}) and $\lambda = 2\alpha y + \beta$ (by \eqref{eq:lem_bisec_y}). Combining these two equalities with \eqref{eq:lem_bisec_w}, KKT conditions \eqref{eq:lem_bisec_w}-\eqref{eq:lem_bisec_comp} can be reduced as 
\begin{equation}\label{eq:lem2_wi}
\mathbf{\bar{w}}_i=-\frac{\mathbf{\bar{b}}_iy}{2\lambda_iy+2\alpha y + \beta},\forall i, \  -\frac{\beta}{2\alpha} \leq y =\| \mathbf{\bar{w}} \| .
\end{equation}
Combining the two equalities in \eqref{eq:lem2_wi}, we have 
\begin{equation}\label{eq:lem_bisec_rooty}
\begin{aligned}
f(y)&\triangleq\sum_{i=1}^M\frac{|\mathbf{\bar{b}}_i|^2}{(2\lambda_iy+2\alpha y + \beta)^2}=1,\\
y&\geq \max\{0,  -\frac{\beta}{2\alpha}\}=y_\textrm{min}.
\end{aligned}
\end{equation}
Since $\lambda_i>0,\forall i$, and $\lambda=2\alpha y + \beta>0$, we have $2\lambda_iy+2\alpha y + \beta> 0$ and hence $f(y)$ is a monotonically decreasing function. Notice $\lambda>0$ holds only when $\|\bm{\Lambda}^{-1}\mathbf{\bar{b}}\| >-\frac{\beta}{\alpha}$ and $\beta<\| \mathbf{\bar{b}}\|$, \textcolor{black}{which implies $f(0)>1$ if $y_\textrm{min}=0$ and $f(-\frac{\beta}{2\alpha})>1$ if $y_\textrm{min}=-\frac{\beta}{2\alpha}$.} Hence we have $f(y_\textrm{min})>1$, and  there is a unique root $y^*>y_\textrm{min}$ satisfying $f(y^*)=1$. On the other hand, consider $f(y^*)=1$, we have 
\begin{equation}
\begin{aligned}
y^*&=y^*\sqrt{f(y^*)}=\sqrt{\sum_{i=1}^M\frac{|\mathbf{\bar{b}}_i|^2}{(2\lambda_i+2\alpha  + \beta/y^*)^2}}\\
&\leq \sqrt{\sum_{i=1}^M\frac{|\mathbf{\bar{b}}_i|^2}{(2\lambda_i)^2}}=\| \bm{\Lambda}^{-1}\mathbf{\bar{b}}\|/2\triangleq y_{\textrm{max}}.
\end{aligned}
\end{equation}
The unique $y^*$ must lie in $[y_{\textrm{min}},y_{\textrm{max}}]$, and can be obtained by bisection search via \eqref{eq:lem_bisec_rooty}. 

Combining the cases for $\beta\geq \|\mathbf{\bar{b}} \|$ and $\beta< \|\mathbf{\bar{b}} \|$, we obtain the solution for problem (19) as follows,
\begin{align*}
\left\lbrace
\begin{aligned}
&y^*=0,\mathbf{w}^*=\mathbf{0},\  \textrm{if }\beta\geq \| \mathbf{b}\|,&\\
&y^*=-\frac{\beta}{2\alpha},\mathbf{w}^*=\mathbf{w}(y^*),\  \textrm{if }\beta< \|\mathbf{b}\|,\| \mathbf{A}^{-1} \mathbf{b}\|\leq -\frac{\beta}{\alpha},&\\
&f(y^*)=1, \mathbf{w}^*=\mathbf{w}(y^*),\ \textrm{otherwise},
\end{aligned}
\right.
\end{align*}
where $\mathbf{w}(y)=-\mathbf{U}^H\left[  2\bm{\Lambda}+(2\alpha  + \frac{\beta}{y})\mathbf{I} \right]^{-1}\mathbf{U}\mathbf{b}$ (by \eqref{eq:lem2_wi}). This completes the proof.
\begin{prop}\label{prop:suff_lem}
	The optimal $(\mathbf{\bar{w}}^*,y^*)$ for \eqref{eq:lem_sub_wy} is $(\mathbf{0},0)$ if and only if  $\beta\geq \| \mathbf{\bar{b}} \|$.
\end{prop}
\begin{proof}
	\textcolor{black}{
	Let the objective function of problem \eqref{eq:lem_sub_wy} be $f(\mathbf{\bar{w}},y)$. We first prove `if' case by showing $f(\mathbf{\bar{w}},y)\geq f(\mathbf{0},0)=0$ for any feasible $(\mathbf{\bar{w}},y)$ if $\beta\geq \| \mathbf{\bar{b}} \|$. Since $\alpha>0$ and $\beta\geq0$, we have  $y^*_{\mathbf{\bar{w}}}=\arg\min_{\| \mathbf{\bar{w}} \| \leq y}f(\mathbf{\bar{w}},y)=\|\mathbf{\bar{w}}\|$. Define $\mathbf{\bar{w}}^\prime=\frac{\mathbf{\bar{w}}}{\| \mathbf{\bar{w}}\|}$ and $\epsilon=\|\mathbf{\bar{w}}\|\geq0$, we have 
	\begin{align*}
	f(\mathbf{\bar{w}},y)&\geq f(\epsilon\mathbf{\bar{w}}^\prime,\epsilon)=(\mathbf{\bar{b}}^H\mathbf{\bar{w}}^\prime/\| \mathbf{\bar{b}}\|+\beta)\epsilon+ |o(\epsilon^2)|\\
	&\geq(-\| \mathbf{\bar{b}}\|+\beta)\epsilon+ |o(\epsilon^2)|\\
	&\geq 0 = f(\mathbf{0},0),
	\end{align*}
	where the second inequality is due to Cauchy-Schwartz inequality $\mathbf{\bar{b}}^H\mathbf{\bar{w}}^\prime\geq-\| \mathbf{\bar{b}}\|^2$, and the third inequality is due to $\beta\geq\| \mathbf{\bar{b}} \|$. 
	Next, we prove `only if' case by contradiction. Suppose $(\mathbf{\bar{w}}^*,y^*)=(\mathbf{0},0)$ is optimal for problem \eqref{eq:lem_sub_wy} and $\beta<\| \mathbf{\bar{b}}\|$. Then, define a feasible solution $(-\epsilon\mathbf{\bar{b}}/\| \mathbf{\bar{b}}\|,\epsilon)$ with $\epsilon>0$, we have $f(-\epsilon\mathbf{\bar{b}}/\| \mathbf{\bar{b}}\|,\epsilon)=(-\| \mathbf{\bar{b}}\|+\beta)\epsilon+ |o(\epsilon^2)|$. Since $\beta<\| \mathbf{\bar{b}}\|$ implies  $(-\| \mathbf{\bar{b}}\|+\beta)\epsilon<0$. For sufficiently small $\epsilon>0$, we have  $f(-\epsilon\mathbf{\bar{b}}/\| \mathbf{\bar{b}}\|,\epsilon)<0$ which contradicts the optimality of $\mathbf{\bar{w}}^*$. This completes the proof.}

\end{proof}

\section{Proof of Proposition 5}\label{app:t}
\begin{proof}
	By Lemma 3, $\partial f(t)$ can be rewritten as 
	\begin{equation}\label{eq:partial_ft}
	\partial f(t)= \sum_k\sum_{\phi\in\Phi_k}\min\{ 2a_{\phi,1}t+b_{\phi,1},2a_{\phi,2}t+b_{\phi,2},0  \}.
	\end{equation}
	By (27) and \eqref{eq:partial_ft}, we know $\partial f(t)$ is a strictly increasing and continuous function for $t\leq \max_\phi\{ \bar{t}_{\phi,2}   \}\triangleq t_{\textrm{max}}$ and $\partial f(t)=0$ for $t\geq t_{\textrm{max}}$. By the first order optimality condition (gradient at optimal $t^*$ is $-\mu$), problem (26) must have a unique optimal solution $t^*\leq t_{\textrm{max}}$ such that $\partial f(t^*)=-\mu$. Notice $\partial f(t)$ is a piecewise linear function, the optimal $t^*$ lies on some region with $\partial f(t)=2at + b$ and $2at^*+b=-\mu$, where $a>0, b\in\mathbb{R}$ are coefficients to be determined. By \eqref{eq:partial_ft}, if $\partial f(t_{\textrm{min}})>-\mu,t_{\textrm{min}}\triangleq\min_\phi\{  \bar{t}_{\phi,1}\}$, then the region must be $t\leq t_{\textrm{min}}$ and hence $a=\sum_k\sum_{\phi\in\Phi_k}a_{\phi,1}$ and $b=\sum_k\sum_{\phi\in\Phi_k}b_{\phi,1}$. Otherwise, there must exist a region within $t_{\textrm{min}}\leq t\leq t_{\textrm{max}}$ such that $\partial f(t)=2at + b$ and $2at^*+b=-\mu$. Sort all $\{ \bar{t}_{\phi,1}\}$ and $\{ \bar{t}_{\phi,2}\}$ as an increasing sequence $\{ \tilde{t}_\ell\}$ and by the strictly increasing property of $\partial f(t)$, the region for $\partial f(t^*)=-\mu$ must be $\tilde{t}_\ell\leq t^* \leq \tilde{t}_{\ell+1}$, where $\partial f(\tilde{t}_\ell) \leq -\mu$ and $\partial f(\tilde{t}_{\ell+1}) \geq -\mu$. In such a case, we have 
	$$a=\sum_{\phi\in\Omega_1} a_{\phi,1}+\sum_{\phi\in\Omega_2}a_{\phi,2},\ b=\sum_{\phi\in\Omega_1} b_{\phi,1}+\sum_{\phi\in\Omega_2}b_{\phi,2},$$
	and 
	$\Omega_1=\{  \phi | \bar{t}_{\phi,1}\geq \tilde{t}_{\ell+1} \},\ \Omega_2 = \{   \phi |  \bar{t}_{\phi,1}\leq \tilde{t}_{\ell}, \bar{t}_{\phi,2}\geq \tilde{t}_{\ell+1} \}.
	$
	Combining the two cases and setting $t^*=-\frac{b+\mu}{2a}$ can complete the proof.
\end{proof}
\end{document}